\documentclass[11pt]{article}

\usepackage[T5]{fontenc}
\usepackage[margin=1in]{geometry}

\usepackage{graphicx} 
\usepackage{algorithm}
\usepackage{algpseudocode}
\usepackage{amsthm,amsmath, mathtools, enumitem, amssymb, thm-restate, bbm}
\usepackage{color, soul}
\usepackage[dvipsnames]{xcolor}

\newcommand{\bool}{\set{0,1}}

\newcommand{\Cbb}{\mathbb{C}}

\newcommand{\Nbb}{\mathbb{N}}

\newcommand{\tb}{\mathbf{t}}

\newcommand{\Ac}{\mathcal{A}}
\newcommand{\Bc}{\mathcal{B}}
\newcommand{\Cc}{\mathcal{C}}
\newcommand{\Dc}{\mathcal{D}}
\newcommand{\Ec}{\mathcal{E}}
\newcommand{\Fc}{\mathcal{F}}

\newcommand{\Sc}{\mathcal{S}}
\newcommand{\Qc}{\mathcal{Q}}

\newcommand{\wt}[1]{\widetilde{#1}}

\newcommand{\eps}{\varepsilon}

\newcommand{\bef}{\preccurlyeq}

\newcommand{\prob}[2][]{\operatorname*{\mathbb{P}}_{#1 }\brac*{#2}}
\newcommand{\Expect}[2][]{\operatorname*{\mathbb{E}}_{#1 }\brac*{#2}}
\newcommand{\Var}[2][]{\operatorname*{\normalfont{\text{Var}}}_{#1 }\paren*{#2}}

\newcommand{\bO}[1]{\operatorname*{O}\paren*{#1}}
\newcommand{\bOt}[1]{\operatorname*{\wt{O}}\paren*{#1}}
\newcommand{\lO}[1]{\operatorname*{o}\paren*{#1}}
\newcommand{\bOm}[1]{\operatorname*{\Omega}\paren*{#1}}

\newcommand{\bT}[1]{\operatorname*{\Theta}\paren*{#1}}

\DeclareMathOperator{\poly}{poly}

\DeclarePairedDelimiter\floor{\lfloor}{\rfloor}
\DeclarePairedDelimiter\ceil{\lceil}{\rceil}
\DeclarePairedDelimiter\norm{\lVert}{\rVert}
\DeclarePairedDelimiter\abs{|}{|}
\DeclarePairedDelimiter\brac{\lbrack}{\rbrack}
\DeclarePairedDelimiter\set{\lbrace}{\rbrace}
\DeclarePairedDelimiter\paren{\lparen}{\rparen}
\DeclarePairedDelimiter\interval{\lbrack}{\rparen}

\DeclarePairedDelimiter\bra{\langle}{|}
\DeclarePairedDelimiter\ket{|}{\rangle}
\DeclarePairedDelimiterX\braket[2]{\langle}{\rangle}{#1\delimsize\vert#2}

\newtheorem{theorem}{Theorem}
\newtheorem{lemma}[theorem]{Lemma}

\theoremstyle{definition}
\newtheorem{definition}[theorem]{Definition}
\theoremstyle{remark}

\newtheorem{fact}[theorem]{Fact}

\newcommand{\nth}{^\text{th}}

\newcommand{\din}{d^\text{in}}
\newcommand{\dout}{d^\text{out}}
\newcommand{\dedge}[1]{\overrightarrow{#1}}
\newcommand{\degb}[1]{d_{#1}^{\rightarrow}}
\newcommand{\tgk}{t^{>k}}
\newcommand{\tlk}{t^{<k}}
\newcommand{\Tgk}{T^{>k}}
\newcommand{\Tlk}{T^{<k}}

\newcommand{\hist}[1]{\mathsf{Snap}^{#1}}
\newcommand{\histps}[1]{\mathsf{PsSnap}^{#1}}
\newcommand{\db}[1]{d^{\leq#1}}
\newcommand{\dbps}[1]{\wt{d}^{\leq#1}}
\newcommand{\da}[1]{d^{>#1}}
\newcommand{\doutb}[1]{d^{\text{out},\leq#1}}
\newcommand{\doutbps}[1]{\wt{d}^{\text{out},\leq#1}}
\newcommand{\douta}[1]{d^{\text{out},>#1}}
\newcommand{\bps}[1]{\wt{b}^{#1}}

\newcommand{\inc}{\mathtt{inc}}
\newcommand{\measure}{\mathtt{query\_edge}}
\newcommand{\cleanup}{\mathtt{cleanup}}

\newcommand{\create}{\mathtt{create}}
\newcommand{\update}{\mathtt{update}}
\newcommand{\queryone}{\mathtt{query\_one}}
\newcommand{\querypair}{\mathtt{query\_pair}}

\newcommand{\yes}{\texttt{Yes}}
\newcommand{\no}{\texttt{No}}

\newcommand{\accuracy}{\kappa}
\newcommand{\prefix}{\rho}

\newcommand{\mdicut}{{\normalfont\textsc{Max-DiCut}}}
\newcommand{\mdcut}{\mdicut}

\newcommand{\examplebox}[1]{\noindent\fbox{%
    \parbox{0.97\textwidth}{%
    \textbf{Example:\\}
    #1
    }
}
}

\title{How to Design a Quantum Streaming Algorithm Without Knowing Anything About Quantum Computing}

\date{}
\author{John Kallaugher\\Sandia National
Laboratories\\\texttt{jmkall@sandia.gov} \and Ojas Parekh\\Sandia National
Laboratories\\\texttt{odparek@sandia.gov} \and Nadezhda
Voronova\\ Boston University\\\texttt{voronova@bu.edu}}

\begin{document}
\maketitle
\begin{abstract}
\noindent
A series of work~\cite{GKKRW07,K21,KPV24} has shown that asymptotic advantages
in \emph{space complexity} are possible for quantum algorithms over their
classical counterparts in the streaming model. We give a simple quantum sketch
that encompasses all these results, allowing them to be derived from entirely
classical algorithms using our quantum sketch as a black box. The quantum
sketch and its proof of correctness are designed to be accessible to a reader
with no background in quantum computation, relying on only a small number of
self-contained quantum postulates.
\end{abstract}

\newpage
\tableofcontents
\newpage

\pagenumbering{arabic}

\section{Introduction}
Settings in which data is generated much faster than it can be stored are
increasingly commonplace in many applications, such as those arising
in processing and analyzing internet traffic.  In the data stream model, which
captures such settings, random access to the input is prohibited, and input
data elements arrive sequentially.  A streaming algorithm in this model must
process each data element as it arrives using as little space as possible,
ideally polylogarithmic in the size of the entire input.  While streaming
algorithms were conceived to compute statistics of numerical data
streams~\cite{Mor78,FM85,AMS96}, they have since been applied more
broadly---notably for estimating graph
parameters~\cite{BKS02,McGregorSurvey14}.  Streaming graph algorithms are well
poised to address applications stemming from massive graphs, such as those
derived from social or sensor network activity.

Quantum computing leverages quantum mechanics to process information and offers
the hope of exponential resource advantages over classical computing.  Certain
problems, such as factoring integers, can be solved exponentially faster with
quantum algorithms than with the best-known classical algorithms~\cite{S99}.
Yet provable exponential quantum advantages against best-possible classical
algorithms remain scarce.  Provable exponential advantages are achievable in
restricted computational models such as the data stream
model~\cite{LG06,GKKRW07}, where such advantages are with respect to space (the
number of classical or quantum bits required) rather than execution time.  

Space is an especially critical resource for quantum computing, as developing
error-tolerant and scalable quantum bits (qubits) is a major challenge.  It is
unclear whether noisy intermediate-scale quantum (NISQ) computers or early
fault-tolerant quantum computers, with relatively limited quantities of logical
qubits, will be able to realize quantum advantages.  Space-efficient quantum
algorithms, including quantum streaming algorithms, offer an alternative
opportunity for such advantages.  Even early quantum computers with limited
quantum memory may be able to process larger data sets than possible
classically, and problems admitting exponential quantum space advantages are
promising candidates.   
\paragraph{The Quantum Streaming Model} In the streaming model, the input to a
problem is received as a ``stream'' $(\sigma_i)_{i=1}^m$, one element at a
time, in an arbitrary order. For each possible value of $\set{\sigma_i : i \in
\brac{m}}$ there is a set of valid answers to output. The task of a
streaming algorithm is to output one of these valid answers after processing
the stream.  Note that the answer is not allowed to depend on the \emph{order}
of the stream, and a ``promise'' on the input can be enforced by allowing all
possible answers for inputs violating the promise. For instance, in a graph
streaming problem, the stream elements $\sigma_i$ might be edges of the graph,
and the desired output might be to approximate some parameter of the graph to a
desired accuracy.  Typically in this model we are interested in the
\emph{space} complexity of the algorithm: how many bits (or qubits) of storage
are needed to solve the problem, although other parameters, such as update
time, may also be considered.

In the \emph{quantum} streaming model, the input stream takes the same
(classical) form, but it is processed by an algorithm with access to quantum
resources. Formally, we may think of the algorithm as having an array of qubits
and a classical algorithm that, for every update processed, chooses a quantum
circuit built from some finite set of gates and measurement operators, to apply
to the qubits. The algorithm is allowed to perform both pre-processing and
post-processing operations. Its space complexity, then, is the sum of the
number of qubits used and the number of bits required by the classical
algorithm. Note that we allow measurements at intermediate points while
processing the stream: it is not in general possible to defer measurements to
the end of the algorithm because \emph{which} measurements are performed may
depend on the updates seen, which would require additional space to remember.

\paragraph{Sketching} A closely related concept is the idea of a ``sketch'' of
a dataset~\cite{Cor17}. This is a succinct representation of a dataset (smaller than the
dataset itself; for our purposes we will want it to use at most $\lO{n}$ bits or
qubits for an $n$-bit input) that supports a restricted set of queries. Basic
results in classical and quantum information imply that these queries cannot be
expressive enough to allow reconstructing the entire dataset, and we will also
allow them to be \emph{randomized}, obeying only some probabilistic guarantee
of correctness. We will often refer to a sketch as ``containing'' an element or
having an ``element'' added to it, to denote that element being in the dataset
sketched, but it is important to note that this does not put sketches in
one-to-one correspondence with sets of elements---two identical sketches may
contain different elements in this sense.

A minimal requirement for a sketch to be useful for streaming is that it can be
updated: whenever a new element is seen in the stream, it must be possible to
update the sketch to contain that element. Often other update operations, such
as the ability to merge sketches, are considered---in our case, we will
consider sketches that support permuting the underlying dataset (that is, given
a permutation $\pi$, there is an update operation for the sketch that is
equivalent to applying $\pi$ to every element previously added, \emph{before}
it was added to the sketch).

In the quantum setting, the no-cloning theorem~\cite[p.~79]{Wil17} means that
it is not, in general, possible to copy a quantum sketch. This in turn makes it
possible for the queries supported by the sketch to be meaningfully
\emph{destructive}, irreversibly changing the sketch. 

\paragraph{Other Quantum Streaming Models}
We only consider the standard one-pass streaming model in this work. When
multiple passes are allowed (that is, the input stream is given to the
algorithm multiple times), quantum advantages are known that are not encompassed
by our sketch, provided the number of passes is more than
constant~\cite{LG06,M16,HM19}.

It is known~\cite{AD11} that quantum automata with only a single qubit can
solve the ``coin problem'' of differentiating a coin that is heads with
probability $p$ from one that is heads with probability $p + \eps$, while a
classical automaton would need $\bOm{(p(1-p)/\eps}$ states~\cite{HC70}.
However, such quantum automata cannot be implemented as constant-space quantum
streaming algorithms under our definition, as they depend on the ability to
perform noiseless rotations that depend on $\eps$, and so accounting for noise
(or the error introduced by constructing such rotations from a finite set of
gates) would blow up the space.

\paragraph{Related Work} A provable exponential quantum space advantage for a
streaming graph problem was first\footnote{A similar separation was proved
in~\cite{LG06}, although the problem considered is not formally a streaming
problem in the sense we use in this paper, as it comes with a guarantee on the
\emph{order} in which the updates will be received.} established by Gavinsky,
Kempe, Kerenidis, Raz, and de Wolf~\cite{GKKRW07}, who considered a streaming
variant of the Boolean Hidden Matching (BHM) problem.  While BHM can be seen as
a graph problem, it is not known to have algorithmic applications, although it
is a powerful tool for proving \emph{lower bounds}.  The question of a quantum
streaming advantage for a ``natural'' and previously classically studied
problem, as articulated by Jain and Nayak~\cite{JN14}, remained open.
Kallaugher discovered the first such advantage, for counting triangles in graph
streams~\cite{K21}, which is well studied and motivated by
applications~\cite{BKS02,BOV13,JK21}.  However, this advantage is polynomial in
the input size, and Kallaugher, Parekh, and Voronova~\cite{KPV24} gave the
first natural exponential quantum streaming advantage, for approximating the
value of the Max Directed Cut problem in a directed graph stream.  This also
represents the first known quantum advantage for approximating a graph problem,
albeit in the streaming model.  The details of the graph problem considered are
critical, since Kallaugher and Parekh~\cite{KP22} showed that no quantum
streaming advantage is possible for the very similar Max Cut problem in
undirected graphs.     

\paragraph{Our Contributions} One of our main goals is to provide experts in
designing classical algorithms a simple entry point for designing novel and
cutting-edge quantum algorithms offering quantum advantages over classical
counterparts.  We observe that quantum streaming algorithms provide a unique
opportunity toward this end.  Learning more traditional quantum algorithms such
as Shor's algorithm~\cite{S99} requires an understanding of quantum circuits
and more basic kernels such as the Quantum Fourier Transform.  The details of
such algorithms can be daunting and a potential deterrent for newcomers to
quantum computing.  In contrast, we show that a simple quantum sketch coupled
with more sophisticated but purely classical techniques captures
state-of-the-art quantum streaming algorithms that represent recent
breakthrough results.  While conventional wisdom might suggest that novel
\emph{quantum} techniques are necessary for breakthroughs in quantum
algorithms, we observe through the present work that the novelty lies on the
classical side for recent quantum streaming algorithms.  More broadly we
believe that effectively enlisting and engaging classical algorithms
researchers is critical in advancing quantum algorithms, and we intend for this
work to provide such a path.  

We show that the quantum streaming advantages for the Boolean Hidden Matching~\cite{GKKRW07}, Triangle Counting~\cite{K21}, and Max Directed Cut~\cite{KPV24} problems may be cast as classical algorithms employing a simple-to-understand quantum sketch as a black box.  The sketch and accompanying proofs of correctness rely only on a few basic and self-contained quantum postulates.  As such, they are accessible to readers with no background in quantum computing and can serve as an alternative to other first examples in quantum information such as quantum teleportation or superdense coding~\cite{NC10}.  Isolating and abstracting the quantum components of the above quantum streaming advantages also better illustrates how and why quantum streaming advantage arises, especially to those unversed in quantum algorithms.  

\paragraph{How to Read This Paper} In the remainder of this section, we will
give a high level description of the sketch and its application to streaming
algorithms. For an exact description of the operations provided by the sketch,
see Section~\ref{sec:description}. In Section~\ref{sec:construction} we prove
that the sketch can be implemented. This is the only part of the paper that
uses any quantum information, and can be skipped by the reader who is happy to
have the sketch as a ``black box''. However, the quantum information postulates
used are all stated in the section, and so it should not require any additional
quantum background.

In Sections~\ref{sec:bhm},~\ref{sec:triangle},~\ref{sec:dicut}, we give explicit
constructions of the known one-pass quantum streaming algorithms as algorithms
that are classical other than their use of the sketch as a subroutine. These
are in increasing order of complexity. For
Sections~\ref{sec:triangle},~\ref{sec:dicut}, the ``quantum part'' of the
original algorithms can in each case be characterized by a single lemma, so we
prove the ``black box'' variant of that lemma rather than reproducing all the
content of the original paper. Section~\ref{sec:dicut} also contains, by way of
presenting a simpler version of the problem solved, a solution to a ``counting heavy
edges'' problem that may be of independent interest as a streaming primitive. 

\subsection{A Quantum Sketch for Pair Sampling}
The core of our paper is a quantum sketch $\Qc_T$ for a set $T$ (contained in a
$\poly(\abs{T})$-sized universe of possible elements $U$) that solves a ``pair
sampling'' problem in the stream. Given a stream of query pairs $(u,v) \in P
\subseteq U^2$, we want to estimate the number of pairs from $P$ that are
contained in $T$ (that is, $\abs{\set{(u,v) \in P : u \in T \wedge v \in T}}$),
and, assuming we determine that this number is non-zero, we want to, by the end
of the stream, randomly sample one of those pairs in $\set{(u,v) \in P : u \in
T \wedge v \in T}$. Note that we do not know $P$ at the time we form the sketch $\Qc_T$.

One classical approach to this problem would be to form $\Qc_T$ by randomly
subsampling $T$. It is easy to see that this would require sampling at least
$\bOm{\sqrt{\abs{T}}}$ elements: if $k$ elements are sampled, any given pair is
lost with probability $1 - k^2/\abs{T}^2$, and so if $k = \lO{\sqrt{\abs{T}}}$,
we can expect to lose \emph{every} pair from $P$ contained in $T$, even if the
greatest possible  number of such pairs ($\abs{T}/2$) is present. A reduction
to the Boolean Hidden Matching problem~\cite{GKKRW07} proves that this is
optimal up to a log factor: \emph{any} classical algorithm needs
$\bOm{\sqrt{\abs{T}}}$ bits of storage to solve the sampling part \emph{or} the
counting part of the problem, even if we only want, say, a $\pm \abs{T}/10$
bound on the additive error.

Our quantum sketch solves this problem with only $\bO{\log n}$ qubits, assuming
the universe $U$ has size $\poly(n)$.  This comes with two important caveats:
\begin{itemize}
\item It only works as described when the pairs in $P$ are \emph{disjoint}.
However, it will not fail arbitrarily when this does not happen. Rather, every
time a pair query $(u,v)$ is made, both $u$ and $v$ are removed from the
sketch, with the effect that e.g.\ a subsequent query to $(u,w)$ will be
treated as if $u$ were not in $T$. Working around this limitation is a core
challenge in writing algorithms that make use of the sketch, as in e.g.\ graph
problems where the pairs we want to query are edges, there will sometimes be
high-degree vertices that cause a large amount of overlap among the pairs to be
queried.
\item When the sketch returns a sample (which is supposed to be a pair from $P$
contained in $T$), it may sometimes instead return an element $(u,v)$ from $P$
such that only \emph{one} of $u$ and $v$ is a member of $T$. However, we will
have two guarantees that mitigate this drawback. Firstly, the probability of
returning a valid pair will be proportional to the fraction of $T$ that is
covered by pairs from $P$, i.e.\ to $\frac{\abs{\set{(u,v) \in P : u \in T
\wedge v \in T}}}{\abs{T}}$.  Secondly, the sample will come with a $\pm$ sign.
If the sample is valid, this will always be $+$, while if it is spurious it
will be $+$ or $-$ with equal probability. We can therefore use the $-$
responses to ``cancel out'' the effect of spurious samples, depending of course
on the application in question.
\end{itemize}
In order to be useful for streaming algorithms, we must be able to update this
sketch as well as query it in the stream. We will have a somewhat more powerful
update rule than simply being able to add elements to the sketch: at any point,
we can execute an arbitrary permutation $\pi$ (on the universe $U$), replacing
the sketched set $T$ with $\pi(T)$. In particular, we can add elements to the
sketch by initializing it on a set of ``dummy elements'' and then using this
operation to replace a dummy element whenever we want a new element added to
the sketch (although note that means we must start by initializing the sketch
with a set as large as the set of all elements we will add to the sketch).

A formal description of the operations supported by the sketch can be found in
Section~\ref{sec:description}. Note that we formalize the query process by 
saying that a query to a pair $(u,v)$ contained in $T$ returns \yes{} with
probability proportional to $1/\abs{T}$ (destroying the set as this happens),
and otherwise causes $u$ and $v$ to be removed from $T$, which will suffice
both for the estimation and sampling part of the task described. In fact, it is
somewhat more powerful, as the fact that we don't have to wait until the end of
the stream to receive our sampled pair will be necessary for the Max Directed
Cut algorithm in Section~\ref{sec:dicut}.

\subsection{Streaming Applications}
We give three applications of the sketch, corresponding to the three known
examples of one-pass quantum advantage in the streaming model.

\paragraph{Boolean Hidden Matching} The first demonstration of quantum space
advantage in streaming was in~\cite{GKKRW07}, for a streaming version of the
Boolean Hidden Matching problem (also described in this paper). Boolean Hidden
Matching is a one-way communication problem in which Alice has a string $x \in
\bool^n$ representing set of vertices $[n]$ labeled by bits, and Bob has a
partial matching $M$ on $\brac{n}$ (that is, a set of disjoint edges with
vertex set $\brac{n}$), along with an edge label for each edge in the matching,
given by $z \in \bool^{M}$. There is a secret bit $b \in \bool$ such that $x_u
\oplus x_v \oplus z_{uv} = b$ for every edge $uv \in M$, and their task is to
determine $b$, using as little communication as possible and with only one-way
communication (Alice to Bob). The authors of~\cite{GKKRW07} showed that this
problem can be solved with $\bO{\log n}$ qubits of communication, but requires
$\bOm{\sqrt{n}}$ bits if only classical communication is allowed.

In the streaming version, the stream consists of edges from $M$ with their
labels, and individual bits of $x$ (in the form $(i, x_i)$), in any order and
interspersed in any way.  They gave a $\bO{\log n}$-qubit streaming algorithm
for this problem, based on the same principles as the communication protocol.

We show that this streaming problem can be solved with $\bO{\log n}$ space by
an algorithm that is entirely classical other than its use of our sketch. This
is the most ``pure'' use of the sketch, as the edges in $M$ correspond exactly
to the pairs we want to sample, as the problem can be solved with knowledge of
any $x_u,x_v$ and $z_{uv}$ (indeed, the sketch was inspired by the quantum
protocol for this problem). The main complications come in the fact that an
edge $uv$ may arrive before one or both of $x_u$ and $x_v$. This algorithm is
described in Section~\ref{sec:bhm}.

\paragraph{Triangle Counting} After~\cite{GKKRW07} established that exponential
quantum streaming advantage was possible, the first work to show advantage for
a \emph{natural} one-pass streaming problem was~\cite{K21}, which gave it for
the \emph{triangle counting} problem. This is a graph streaming problem, in
which an $m$-edge graph $G$ is received one edge at a time, and the objective is to
approximate the number of \emph{triangles} in $G$—the number of triples $u,v,w$
such that $uv, vw, wu$ are all edges in the graph. It is known~\cite{BKS02}
that small-space algorithms for this problem are impossible when $G$ contains
very few triangles, or when~\cite{BOV13} too many of these triangles share the
same edge. In this discussion, for simplicity we will assume $G$ contains
$\bT{m}$ triangles, no more than $\bO{1}$ of which intersect at any given edge.

In this setting, by~\cite{KP17}, no classical algorithm can do better than
$\bOm{\sqrt{m}}$ space, with the ``hard instance'' being when every triangle in
the graph intersects at a single vertex. In particular, this occurs when the
stream is a star ($\bT{m}$ edges all incident to one vertex) followed by edges
that may or may not complete triangles with the edges of that star. Conversely,
the quantum algorithm of~\cite{K21} achieves $\bOt{m^{1/3}}$ space. 

We show how to implement that algorithm as a classical algorithm with access to
the quantum sketch as a subroutine. Here the set sketched contains the ordered
pairs $(u,v)$ and $(v,u)$ for each edge $uv$ seen in the stream. Then, the
``pairs'' to be queried whenever we see an edge $uv$ are $((w,u),(w,v))$ for
each $w$ in the set of possible vertices for the graph. Therefore, the number
of ``yes'' answers is the number of triangles completed by $uv$. The
complication in the algorithm comes from the fact that while the pairs queried
for any individual $uv$ are disjoint, there is an overlap between those queried
for any incident edges $uv$, $vt$.

As in the original algorithm of~\cite{K21}, the solution uses the fact that the
worst case for the quantum-sketch-based strategy (i.e.\ for every triangle, its
third edge to arrive in the stream is incident to many other edges) is
incompatible with the worst case for classical algorithms. Therefore, by
interpolating between an entirely classical estimator and an estimator based on
the quantum sketch, an algorithm improving over the best classical algorithm is
achieved. However, the cost of this interpolation is that only a polynomial
advantage ($\bO{m^{1/3}}$ v.\ $\bOm{\sqrt{m}}$) is achieved, despite 
exponential advantage being possible on those instances that are hardest for
classical algorithms. We describe the algorithm in Section~\ref{sec:triangle}.

\paragraph{Maximum Directed Cut} Finally,~\cite{KPV24} gave exponential quantum
streaming advantage for a different natural graph problem: the maximum directed
cut problem in (directed) graph streams. The maximum directed cut value of a
digraph $(V,E)$ is defined as the maximum over all $x \in \bool^V$ of
$\abs{\set{\dedge{uv} \in E : x_u = 0 \wedge x_v = 1}}$. That is, it is the
maximum, over all partitions of $V$ into a ``head'' and ``tail'' set, number of
edges that can be made to go from the head set to the tail set. The objective
of the streaming problem is to approximate this number with the best possible
\emph{approximation} ratio $\alpha$: to return a number in $\brac{\alpha \cdot
\text{Opt}, \text{Opt}}$, where $\text{Opt}$ is the true optimal cut value.

In~\cite{CGV20}, it was shown that any approximation ratio than $\frac{4}{9}$
requires $\bOm{\sqrt{n}}$ classical space, while~\cite{SSSV23b,S23} showed that
$0.4844$ is achievable in $\bOt{\sqrt{n}}$ space by estimating a ``first-order
snapshot''. This is given by grouping the vertices of the graph into a constant
number of classes $C_i$ based on their biases (the bias of a vertex $v$ is $b_v =
\frac{\dout_v - \din_v}{d_v}$, where $\dout_v$ is the number of edges
$\dedge{vu}$ in the graph, $\din_v$ is the number of edges $\dedge{uv}$, and
$d_v = \dout_v + \din_v$), and counting, for every pair of classes $C_i, C_j$,
the number of edges $\dedge{uv}$ with $u \in C_i$ and $v \in C_j$.

This problem has a natural ``pair sampling'' structure: for each pair of
classes we want to sketch the vertices in $C_i$ and those in $C_j$, and then,
for each edge $\dedge{uv}$, query the sketch to ask if $u$ is in $C_i$ and $v$
is in $C_j$. There are two major obstacles to achieving this. One is the issue
we saw in the triangle counting case: the queries we want to make will not be
disjoint. The other is that we need to dynamically update the ``bias class'' of
each vertex as the stream is processed, as these classes depend on the edges
that have been seen incident to the vertex.

It turns out that the same solution applies to both of these issues: each
vertex is represented by a ``stack''  of elements in the sketch (in fact,
several stacks) $(v,i)_{i = 1, \dots}$ which is incremented whenever an edge is
seen incident to that sketch. Then, the presence of an element $(v,d)$
corresponds to the vertex having degree at least $d$, and so querying e.g.\
$((u,d),(v,d))$ tests whether vertices $u$ and $v$ both have degree $\ge d$ in
the stream seen so far. This avoids the problem of pairs overlapping, as while
this query removes $(u,d)$ and $(v,d)$ from the sketch, it also replaces them
when the stack is incremented. These degrees are then used to calculate biases
(some additional tricks are required to encode both the in- and out-degree of a
vertex in one stack, and several stacks are needed per vertex to account for
different degree thresholds that correspond to different biases).

This is the logic of the algorithm in~\cite{KPV24}—in Section~\ref{sec:dicut}
we show how it can be implemented with an entirely classical algorithm that
runs the quantum sketch as a subroutine.

\section{Quantum Pair Sketch}
\label{sec:description}
In this section we will describe a quantum sketch $\Qc_T$ that summarizes a set
$T$ contained in a known universe $U = \set{0,1}^m$. It will require $m = \log
|U|$ qubits to store and support both updates to $T$ and queries that check
(with probability $1/\abs{T}$) whether a given element or pair is contained in
$T$.  This operation will be \emph{destructive}: if the queried element is
successfully identified as being in $T$, the sketch is destroyed, and otherwise
it is removed from $T$.

Formally, the operations are as follows:
\begin{itemize}
    \item $\create(T)$: Takes as input $T \subseteq U$ and returns the sketch $\Qc_T$.
    \item $\update(\pi,\Qc_T)$: Takes as input a permutation $\pi : U \to U$ and
    a sketch $\Qc_T$, and (in-place) transforms $\Qc_T$ to $\Qc_{\pi(T)}$,
    where $\pi(T) = \set{\pi(x) : x \in T}$.
    \item $\queryone(x,\Qc_T)$. Takes as input an element $x \in U$ and a
    sketch $\Qc_T$, and probabilistically tests whether $x \in T$. Returns an
    element of $\set{\in, \bot}$ and modifies the sketch as follows:
    \begin{itemize}
      \item If $x \in T$:
        \begin{itemize}
          \item[]With probability $1/\abs{T}$ destroy $\Qc_T$ and return $\in$.
          \item[]With probability $1 - 1/\abs{T}$ replace $\Qc_T$ with
            $\Qc_{T \setminus \set{x}}$ and return $\bot$.
        \end{itemize}
      \item If $x \not\in T$:
        \begin{itemize}
            \item[] Leave $\Qc_T$ unchanged and return $\bot$.
        \end{itemize}
    \end{itemize}
     \item $\querypair(x, y, \Qc_T)$: Takes as input elements $x\not=y$ from $U$
     and a sketch $\Qc_T$, and probabilistically tests whether $\set{x,y}
     \subseteq T$. Returns an element of $\set{+1,-1,\bot}$ and modifies the
     sketch as follows:
     \begin{itemize}
        \item If $\set{x,y} \subseteq T$:
            \begin{itemize}
                \item[] With probability $2/\abs{T}$ destroy $\Qc_T$ and return $+1$.
                \item[] With probability $1 - 2/\abs{T}$ replace $\Qc_T$ with
                $\Qc_{T \setminus \set{x,y}}$ and return $\bot$.
            \end{itemize}
        \item If $\abs{\set{x,y} \cap T} = 1$:
            \begin{itemize}
                \item[] With probability $1/\abs{T}$ destroy $\Qc_T$ and return
                a randomly chosen $r \in \set{+1,-1}$.
                \item[] With probability $1 - 1/\abs{T}$ replace $\Qc_T$ with
                $\Qc_{T \setminus \set{x,y}}$ and return $\bot$.
            \end{itemize}
        \item If $\set{x,y} \cap T = \emptyset$:
           \begin{itemize}
                \item[] Leave $\Qc_T$ unchanged and return $\bot$.
           \end{itemize}
     \end{itemize}
\end{itemize}

We note that, because of the no-cloning theorem~\cite[p.~79]{Wil17}, it is not
possible to copy the sketch, or to check whether it contains some element
without potentially destroying the sketch. 

It may be unclear how the output of the query operations above may be construed
as \yes{} or \no{} answers to queries.  Any query that returns an answer
besides $\bot$ is considered a success, and destroys the sketch.  A successful
answer for $\queryone$ is $\in$, which is interpreted as \yes{}.  The
$\querypair$ operation returns numerical values upon success.  By designing
algorithms that consider the numerical values in \emph{expectation}, we may
interpret $+1$ as \yes{} and $0$ as \no{}.  This is because when
$\querypair(x,y,\Qc_T)$ succeeds and $\{x,y\} \subseteq T$, then it always
returns $+1$.  If $\querypair(x,y,\Qc_T)$ succeeds and $|\{x,y\} \cap T| = 1$,
the values $\pm 1$ are returned with equal probability, and so they cancel out
to 0 in expectation, which we interpret as \no{}.  This will suffice for many
applications.

We will make frequent use of an additional operation that relies only on the
basic sketch operations defined above.  We will want to add elements to $T$
during the course of the sketch.  If we have a bound $m$ on the maximum number
of elements we will ever want to add, then we add $m$ extra ``dummy'' elements
$\{d_1,\ldots,d_m\}$ to the universe $U$, and include those elements in the set
$T$ used when we initialize the sketch with $\create(T)$.  The $t\nth$ time we
want to add a new element $x$ to $T$, we perform $\update(\pi,T)$ taking $\pi$
to be the transposition that swaps $d_t$ and $x$. These dummy elements are not
quite ``free''---the size of the set $T$  will depend on $m$, and this will
impact the outcome probabilities for query operations on $\Qc_T$.

One more useful property of this sketch, which we will prove in the following
section, follows as direct consequence of the definitions above.
\begin{restatable}{lemma}{reordering}
\label{lm:reordering}
For any sequence of $\queryone$ and $\querypair$ operations applied to $\Qc_T$, interleaved in any
way with $\update$ operations, there is a unique $T'$ such that the sketch will
become $\Qc_{T'}$ after these operations \emph{if none of them destroy the
sketch}. Moreover, the probability that none of the operations destroy the
sketch is $\frac{\abs{T'}}{\abs{T}}$.
\end{restatable}
Note that in particular, this implies that if we perform any sequence of
$\queryone$ and $\querypair$ operations on \emph{disjoint} singletons and pairs,
the distribution on outcomes is independent of the ordering of the operations,
as the probability of each non-$\bot$ outcome of each query is proportional to
the size of the set underlying the query at the time the query is made. This
means that when we conduct a batch of disjoint queries it will not be necessary
to worry about the order in which they are made, simplifying the description of
our algorithms.

\section{Implementing the Pair Sketch}
\label{sec:construction}
In this section we will show how to implement the sketch described in
Section~\ref{sec:description}.  This section is self-contained, presented in an elementary 
manner, and does not require any quantum computing background.  However, this section
may be skipped, since the quantum streaming algorithms in the sequel are purely classical beyond
using the quantum pair sketch in a black-box manner that does not require understanding its implementation.
\begin{theorem}
\label{thm:implementability}
A quantum algorithm with $\log{\abs{U}}$ qubits can implement the sketch $\Qc$
described in Section~\ref{sec:description}.
\end{theorem}
\subsection{Quantum Preliminaries}
We start by describing three facts about quantum computing that suffice to
construct the sketch. For a more complete discussion, see e.g.~\cite{NC10}.
In order to make this section more accessible, we use a running example in the notation of traditional linear algebra along with the braket notation.

First, we describe a superposition over qubit states. We will follow the
standard in quantum information of using $\ket{\psi}$ to denote a (column)
vector, and $\bra{\psi}$ to denote its conjugate transpose (i.e., $\bra{\psi} = \ket{\psi}^\dagger$). For a set $S$ and
complex vector space $\Cbb^S$, we write the standard basis elements of the
space as $\ket{s}$ for each $s \in S$.
\begin{fact}[State]
\label{fact:superposition}
An $m$-qubit \emph{state} $\ket{\psi}$ is a unit vector in
$\Cbb^{2^m}$. 
\end{fact}
As the name implies, such a state can be stored using $m$ quantum bits
(qubits). Note that this promises nothing about how we can interact with a
state.  Quantum mechanics does not allow for random access to the entries of a
state $\ket{\psi}$, and extracting information from a state is typically a
destructive process.

\examplebox{
	Let $m = 2$. Then $\ket{\psi}$ represents some column vector 
	\begin{align*}
    \ket{\psi} & = \begin{bmatrix}
    		\alpha_{1} \\
        \alpha_{2} \\
        \alpha_{3} \\
        \alpha_{4} \\
    \end{bmatrix} \in \Cbb^{4}, \norm{\ket{\psi}}_2 = 1,
  \end{align*}

which requires $2$ qubits to store. The basis elements are denoted as 

	\begin{align*}
    \ket{1} & = \begin{bmatrix}
    		1 \\
        0 \\
        0 \\
        0 \\
    \end{bmatrix} 
    & \ket{2} & = \begin{bmatrix}
    		0 \\
        1 \\
        0 \\
        0 \\
    \end{bmatrix}
    & \ket{3} & = \begin{bmatrix}
    		0 \\
        0 \\
        1 \\
        0 \\
    \end{bmatrix}
    & \ket{4} & = \begin{bmatrix}
    		0 \\
        0 \\
        0 \\
        1 \\
    \end{bmatrix}
  \end{align*}

}

Next we will describe how states can be computed.

\begin{fact}[Quantum Computation]
\label{fact:computation}
A \emph{quantum computation} starts with a fixed $m$-qubit state, typically
$\ket{1}$ as defined in the example above, and applies a finite sequence of
unitary matrices on $\Cbb^{2^m}$ (linear transformations that preserve inner
products and therefore vector lengths) to produce a state $\ket{\psi}$: 
\begin{equation*}
\ket{\psi} = U_l U_{l-1} \cdots U_1 \ket{1}.
\end{equation*}  
Each $U_i$ is selected from a finite set of \emph{gates}, and $l$ is the length or running time of the computation.
\end{fact}

  Since $U = U_l \cdots U_1$ is a unitary matrix, any quantum computation can be viewed as applying a unitary transformation to $\ket{1}$.  However, since there are any infinite number of unitaries, arbitrary unitaries cannot be implemented by quantum computation in finite time.  By the same token, arbitrary states cannot be produced in finite time.  As in the classical case, quantum computation that is polynomial in length with respect to the number of qubits is of particular interest as a proxy for tractable computation on a quantum computer.   

Finally, we will describe how we can extract classical information from a state using \emph{projective}\footnote{This is only one way of formalizing quantum measurement, but it will suffice for our purposes.}  measurements.
\begin{fact}[Projective Measurement]
\label{fact:measurement}

Let $(P_i)_{i \in I}$ be a collection of orthogonal projectors on $\Cbb^{2^m}$ (i.e., $P_i^2 = P_i = P_i^\dagger$), such that $\sum_{i \in I}
P_i = I_{2^m}$.  We interpret $i \in I$ as probabilistic outcomes, which we will call \emph{measurement outcomes}.  We can \emph{measure} a state $\ket{\psi}$ with these
projectors, with the following distribution on outcomes:
\begin{itemize}
\item[] For each $i \in I$, with probability $\norm{P_i\ket{\psi}}_2^2$, return
the measurement outcome $i$, and replace the state $\ket{\psi}$ with $P_i
\ket{\psi}/\norm{P_i\ket{\psi}}_2$.
\end{itemize}
\end{fact}

\examplebox{
		Let $\ket{\psi} = (\alpha_{1}, \alpha_{2}, \alpha_{3}, \alpha_{4})^{T} \in \Cbb^{4}, \norm{\ket{\psi}}_2 = 1$. Let $P_{(1, 2, 4)}$ be a projective operator onto the subspace generated by $\ket{1}, \ket{2}, \ket{4}$, and $P_{(3)}$ be a projective operator onto the subspace generated by $\ket{3}$: 

		\begin{align*}
    P_{(1, 2, 4)} = \ket{1}\bra{1} + \ket{2}\bra{2} + \ket{4}\bra{4} & = \begin{bmatrix}
    		1 & 0 & 0 & 0\\
        0 & 1 & 0 & 0\\
        0 & 0 & 0 & 0\\
        0 & 0 & 0 & 1\\
    \end{bmatrix} 
    & P_{(3)} = \ket{3}\bra{3} & = \begin{bmatrix}
    		0 & 0 & 0 & 0\\
        0 & 0 & 0 & 0\\
        0 & 0 & 1 & 0\\
        0 & 0 & 0 & 0\\
    \end{bmatrix} 
  \end{align*}

Note that $P_{(1, 2, 4)} + P_{(3)} = I_{4}$. Then the possible outcomes of the measurement of $\ket{\psi}$ are the following:

\begin{itemize}
	\item $(1, 2, 4)$ outcome with probability $$\norm{P_{(1, 2, 4)}\ket{\psi}}_2^2 = \norm{(\alpha_{1}, \alpha_{2}, 0, \alpha_{4})^{T}}_2^2 = |\alpha_{1}|^2 + |\alpha_{2}|^2 + |\alpha_{4}|^2\text{.}$$ 
	The state is updated to be $(\alpha_{1}, \alpha_{2}, 0, \alpha_{4})^{T} / \sqrt{|\alpha_{1}|^2 + |\alpha_{2}|^2 + |\alpha_{4}|^2}$.
	\item $(3)$ outcome with probability $$\norm{P_{(3)}\ket{\psi}}_2^2 = \norm{(0, 0, \alpha_{3}, 0)^{T}}_2^2 = |\alpha_{3}|^2 \text{.}$$ 
	The state is updated to be $ (0, 0, \alpha_{3}/|\alpha_{3}|, 0)^{T} = (0, 0, 1, 0)^{T}$.
\end{itemize}

}

\subsection{Implementation}
For any $T \subseteq U = \brac{2^m}$, we will implement the sketch $\Qc_T$ as
the superposition $\ket{\psi} = \frac{1}{\sqrt{|T|}} \sum_{t \in T} \ket{t}$ in
$\Cbb^U$.  By Fact~\ref{fact:superposition} we can store the state created by
$\create(T)$ with $m = \log{\abs{U}}$ qubits.  For the remainder of this
section, we will write $\ket{\psi_T}$ for this state.  The state $\ket{\psi_T}$
can be efficiently produced by a quantum computation, as dictated by
Fact~\ref{fact:computation} (see e.g., Lemma 1 in~\cite{fefferman2015power}
and~\cite{shukla2024efficient}).  

\examplebox{
	Let $U = \brac{4}$ and $T = \set{1, 2, 3}$. Then the state representing set $T$ is 

	\begin{align*}
    \ket{\psi_T} & = \frac{1}{\sqrt{|T|}} \sum_{t \in T} \ket{t} = \frac{1}{\sqrt{3}}\begin{bmatrix}
    		1 \\
        1 \\
        1 \\
        0 \\
    \end{bmatrix} \in \Cbb^{4}, \norm{\ket{\psi_T}}_2 = 1 \text{.}
  \end{align*}

}

Therefore, we need to prove the operations $\update$, $\queryone$, and
$\querypair$ can be realized on a quantum computer.  We start with $\update(\pi,
\Qc_T)$, recalling its definition: 
\begin{itemize}
\item $\update(\pi,\Qc_T)$: Takes as input a permutation $\pi : U \to U$ and
    a sketch $\Qc_T$, and (in-place) transforms $\Qc_T$ to $\Qc_{\pi(T)}$,
	    where $\pi(T) = \set{\pi(x) : x \in T}$.
\end{itemize}
\begin{lemma}
$\update(\pi,T)$ is realizable on a quantum computer using $m$ qubits.
\end{lemma}
\begin{proof}
Let $V$ be the linear transformation defined by acting on each standard basis
element $\ket{u}$, for $u \in U$, as $V\ket{u} = \ket{\pi(u)}$. Note that this
is unitary, as it preserves vector lengths. Therefore, by
Fact~\ref{fact:computation}, since $\ket{\psi_T}$ is produced by a quantum
computation, a quantum computer may apply $V$ to $\ket{\psi_T}$, resulting in
\[
V\ket{\psi_T} = \frac{1}{\sqrt{|T|}} \sum_{t \in T} V\ket{t} = \frac{1}{\sqrt{|T|}} \sum_{t \in
T} \ket{\pi(t)} = \ket{\psi_{\pi(T)}}
\]
and so $\Qc_{T}$ has been replaced with $\Qc_{\pi(T)}$.
\end{proof}

The $\update$ operations  we will employ for the quantum streaming algorithms
described in this work are each implementable using constant time on a quantum
computer.

\examplebox{
	Let $\pi: [4] \to [4]$ be defined as $\pi(1) = 2, \pi(2) = 3, \pi(3) = 4, \pi(4) = 1$. Then the linear transformation corresponding to $\pi$ is

	\begin{align*}
    U & = \begin{bmatrix}
    		0 & 0 & 0 & 1\\
        1 & 0 & 0 & 0\\
        0 & 1 & 0 & 0\\
        0 & 0 & 1 & 0\\
    \end{bmatrix} 
  \end{align*}

  and, if applied to our $\ket{\psi_T} = \frac{1}{\sqrt{3}} (1, 1, 1, 0)^T$, it will update the state to 

\begin{align*}
    U\ket{\psi_T} & = \frac{1}{\sqrt{3}}\begin{bmatrix}
    		0 & 0 & 0 & 1\\
        1 & 0 & 0 & 0\\
        0 & 1 & 0 & 0\\
        0 & 0 & 1 & 0\\
    \end{bmatrix}
    \begin{bmatrix}
    		1\\
        1\\
        1\\
        0\\
    \end{bmatrix} 
    = \frac{1}{\sqrt{3}} \begin{bmatrix}
    		0\\
        1\\
        1\\
        1\\
    \end{bmatrix}
    = \ket{\psi_{\set{2, 3, 4}}} = \ket{\psi_{\pi(T)}} \text{.}
  \end{align*}

}

\noindent
Next, we show how to implement $\queryone(x, \Qc_T)$:
\begin{itemize}
\item $\queryone(x,\Qc_T)$. Takes as input an element $x \in U$ and a
    sketch $\Qc_T$, and probabilistically tests whether $x \in T$. Returns an
    element of $\set{\in, \bot}$ and modifies the sketch as follows:
    \begin{itemize}
      \item If $x \in T$:
        \begin{itemize}
          \item[]With probability $1/\abs{T}$ destroy $\Qc_T$ and return $\in$.
          \item[]With probability $1 - 1/\abs{T}$ replace $\Qc_T$ with
            $\Qc_{T \setminus \set{x}}$ and return $\bot$.
        \end{itemize}
      \item If $x \not\in T$:
        \begin{itemize}
            \item[] Leave $\Qc_T$ unchanged and return $\bot$.
        \end{itemize}
    \end{itemize}
\end{itemize}

\begin{lemma}
$\queryone(x,\Qc_T)$ is realizable on a quantum computer using $m$ qubits.
\end{lemma}
\begin{proof}
We define projectors $P_{\in}$, $P_{\bot}$ as: 
\begin{align*}
P_{\in} &= \ket{x}\bra{x}\\
P_{\bot} &= I - P_{\in}
\end{align*}
$P_{\in}$ projects onto the basis vector $\ket{x}$, and $P_{\in} + P_{\bot} =
I$, so by Fact~\ref{fact:measurement} we can measure with these projectors. If
$x \in T$, $\norm{P_{\in}\ket{\psi_T}}^2 = \norm*{\frac{\ket{x}}{\sqrt{T}}}^2 =
\frac{1}{T}$ and so we get the outcomes:
\begin{itemize}
\item With probability $\frac{1}{T}$, return $\in$. As we discard the sketch in this case it does not matter what happens to the state.
\item With probability $1 - \frac{1}{T}$, return $\bot$ and $\ket{\psi_T}$ is replaced with $P_{\bot} \ket{\psi_T}$ normalized to length 1, so as \[
P_{\bot} \ket{\psi_T} = \frac{1}{\sqrt{\abs{T}}}\sum_{t \in T \setminus\set{x}} \ket{t} = \frac{\sqrt{\abs{T} - 1}}{\sqrt{\abs{T}}}\ket{\psi_{T \setminus \set{x}}}
\] 
and so after normalization the state is $\ket{\psi_{T \setminus \set{x}}}$ and therefore $\Qc_T$ has been replaced by $\Qc_{T \setminus \set{x}}$.
\end{itemize}
So when $x \in T$ we have the desired distribution on outcomes. When $x \not\in
T$, $\norm{P_{\in}\ket{\psi_T}}^2 = 0$, and so we always return $\bot$ and
$P_\bot\ket{\psi_T} = \ket{\psi_T}$, so the sketch state is always unchanged,
also as desired.
\end{proof}

\examplebox{
	The current state of the sketch in this running example is $\ket{\psi_T} = \frac{1}{\sqrt{3}} (0, 1, 1, 1)^T$, and it stores $T = \set{2, 3, 4}$. 

	Let the query be $\queryone(1, \Qc_T)$. The corresponding projectors are 

		\begin{align*}
    P_{\in} = \ket{1}\bra{1} & = \begin{bmatrix}
    		1 & 0 & 0 & 0\\
        0 & 0 & 0 & 0\\
        0 & 0 & 0 & 0\\
        0 & 0 & 0 & 0\\
    \end{bmatrix} 
    & 
        P_{\bot}  = \ket{2}\bra{2} + \ket{3}\bra{3} + \ket{4}\bra{4} & = \begin{bmatrix}
    		0 & 0 & 0 & 0\\
        0 & 1 & 0 & 0\\
        0 & 0 & 1 & 0\\
        0 & 0 & 0 & 1\\
    \end{bmatrix} 
  \end{align*}

  Then the possible outcomes are 
  \begin{itemize}
  	\item $\in$
  	\begin{itemize}
  		\item[] With probability $\norm{P_{\in}\ket{\psi_{T}}}_2^2 = \norm{(0, 0, 0, 0)^{T}}_2^2 = 0$, and the sketch is destroyed.
  	\end{itemize}
  	\item $\bot$
  	\begin{itemize}
  		\item[] With probability $\norm{P_{\bot}\ket{\psi_{T}}}_2^2 = \norm{\frac{1}{\sqrt{3}}(0, 1, 1, 1)^{T}}_2^2 = 1$, and the updated state is $\frac{1}{\sqrt{3}}(0, 1, 1, 1)^{T} = \ket{\psi_{T\setminus \set{1}}} = \ket{\psi_{T}}$.
  	\end{itemize}
  \end{itemize}
So our state remains unchanged.  If the next query is $\queryone(4, \Qc_T)$, then we have:
		\begin{align*}
    P_{\in} = \ket{4}\bra{4} &     & 
        P_{\bot}  = \ket{1}\bra{1} + \ket{2}\bra{2} + \ket{3}\bra{3} & 
  \end{align*}
The possible outcomes are
\begin{itemize}
  	\item $\in$
  	\begin{itemize}
  		\item[] With probability $\norm{P_{\in}\ket{\psi_{T}}}_2^2 = \norm{\frac{1}{\sqrt{3}}(0, 0, 0, 1)^{T}}_2^2 = 1/3$, and the sketch is destroyed. 
  	\end{itemize}
  	\item $\bot$
  	\begin{itemize}
  		\item[] With probability $\norm{P_{\bot}\ket{\psi_{T}}}_2^2 = \norm{\frac{1}{\sqrt{3}}(0, 1, 1, 0)^{T}}_2^2 = 2/3$, and the updated state is $\frac{1}{\sqrt{2}}(0, 1, 1, 0)^{T} = \ket{\psi_{T \setminus \set{4}}}$.
  	\end{itemize}
  \end{itemize}
Let's assume the second outcome occurred so that the sketch is not destroyed, and we now have $T = \{2,3\}$.
}

Finally, we show how to implement $\querypair(x,y,\Qc_T)$:
\begin{itemize}
  \item $\querypair(x, y, \Qc_T)$: Takes as input elements $x\not=y$ from $U$
	and a sketch $\Qc_T$, and probabilistically tests whether $\set{x,y}
	\subseteq T$. Returns an element of $\set{+1,-1,\bot}$ and modifies the
	sketch as follows:
	\begin{itemize}
	  \item If $\set{x,y} \subseteq T$:
		\begin{itemize}
		  \item[] With probability $2/\abs{T}$ destroy $\Qc_T$ and return $+1$.
		  \item[] With probability $1 - 2/\abs{T}$ replace $\Qc_T$ with
			$\Qc_{T \setminus \set{x,y}}$ and return $\bot$.
		\end{itemize}
	  \item If $\abs{\set{x,y} \cap T} = 1$:
		\begin{itemize}
		  \item[] With probability $1/\abs{T}$ destroy $\Qc_T$ and return
			a randomly chosen $r \in \set{+1,-1}$.
		  \item[] With probability $1 - 1/\abs{T}$ replace $\Qc_T$ with
			$\Qc_{T \setminus \set{x,y}}$ and return $\bot$.
		\end{itemize}
	  \item If $\set{x,y} \cap T = \emptyset$:
		\begin{itemize}
		  \item[] Leave $\Qc_T$ unchanged and return $\bot$.
		\end{itemize}
	\end{itemize}
\end{itemize}

\begin{lemma}
$\querypair(x,y,\Qc_T)$ is realizable on a quantum computer using $m$ qubits.
\end{lemma}
\begin{proof}
We define projectors $P_{+1}, P_{-1}, P_\bot$ as:
\begin{align*}
P_{+1} &= \frac{1}{2}(\ket{x}+\ket{y})(\bra{x}+\bra{y})\\
P_{-1} &= \frac{1}{2}(\ket{x}-\ket{y})(\bra{x}-\bra{y})\\
P_{\bot} &= I - P_{+1} - P_{-1} 
\end{align*}
$P_{+1}$ projects onto the vector $\frac{\ket{x} + \ket{y}}{\sqrt{2}}$ and
$P_{-1}$ onto $\frac{\ket{x} - \ket{y}}{\sqrt{2}}$, while $P_{+1} + P_{-1} +
P_{\bot} = I$, so by Fact~\ref{fact:measurement} we can measure with these
projectors. We can now break down the outcomes of this measurement by whether
both $x$ and $y$ are in $T$, exactly one of them is, or neither of them is, and
show that the distribution of outcomes is as desired in each case.

Firstly, if $\set{x, y} \subseteq T$, then $\norm{P_{+1} \ket{\psi_T}}^2 = 2/T$ and $\norm{P_{-1} \ket{\psi_T}}^2 = 0$, so the outcomes are:
\begin{itemize}
\item With probability $\frac{2}{T}$, return $+1$. As we discard the sketch in
this case it does not matter what happens to the state.
\item With probability $1 - \frac{2}{T}$, return $\bot$ and $\ket{\psi_T}$ is
replaced with $P_{\bot} \ket{\psi_T}$ normalized to length 1, so as \[
P_{\bot} \ket{\psi_T} = \frac{1}{\sqrt{\abs{T}}}\sum_{t \in T \setminus\set{x,y}} \ket{t} = \frac{\sqrt{\abs{T} - 2}}{\sqrt{\abs{T}}}\ket{\psi_{T \setminus \set{x,y}}}
\]
and so after normalization the state is $\ket{\psi_{T \setminus \set{x,y}}}$
and therefore $\Qc_T$ has been replaced by $\Qc_{T \setminus \set{x,y}}$.
\end{itemize}
So we have the right distribution on outcomes when $\set{x, y} \subseteq T$. Next, if
$\abs{\set{x, y} \cap T} = 1$, then $\norm{P_{+1} \ket{\psi_T}}^2 = \norm{P_{-1}
\ket{\psi_T}}^2 = 1/2T$, so the outcomes are:
\begin{itemize}
\item With probability $\frac{1}{T}$, return $+1$. As we discard the sketch in
this case it does not matter what happens to the state.
\item With probability $\frac{1}{T}$, return $-1$. As we discard the sketch in
this case it does not matter what happens to the state.
\item With probability $1 - \frac{2}{T}$, return $\bot$ and $\ket{\psi_T}$ is
replaced with $P_{\bot} \ket{\psi_T}$ normalized to length 1, which as in the
previous case is $\ket{\psi_{T \setminus \set{x,y}}}$ and therefore $\Qc_T$ has
been replaced by $\Qc_{T \setminus \set{x,y}}$.
\end{itemize}
So our distribution is right in this case too. Finally, if $\set{x,y} \cap T =
\emptyset$, $\norm{P_{+1} \ket{\psi_T}}^2 = \norm{P_{-1} \ket{\psi_T}}^2 = 0$, so we
always return, and $P_\bot \ket{\psi_T} = \ket{\psi_T}$ and the sketch is
always unchanged, as desired.
\end{proof}

\examplebox{
	The current state of the sketch in this running example is $\ket{\psi_T} = \frac{1}{\sqrt{2}} (0, 1, 1, 0)^T$, and it stores $T = \set{2, 3}$.  To illustrate additional cases, we will also consider $\ket{\psi_S} = \frac{1}{\sqrt{3}} (1, 1, 0, 1)^T$, corresponding to $S = \{1,2,4\}$.

	Let the queries be $\querypair(2, 3, \Qc_T)$ and $\querypair(2, 3, \Qc_S)$. For both the corresponding projectors are 
	\begin{align*}
    P_{+1} & = \frac{1}{2}\begin{bmatrix}
    		0 & 0 & 0 & 0\\
        0 & 1 & 1 & 0\\
        0 & 1 & 1 & 0\\
        0 & 0 & 0 & 0\\
    \end{bmatrix} 
    & 
    P_{-1} & = \frac{1}{2}\begin{bmatrix}
    		0 & 0 & 0 & 0\\
        0 & 1 & -1 & 0\\
        0 & -1 & 1 & 0\\
        0 & 0 & 0 & 0\\
    \end{bmatrix} 
    & 
    P_{\bot} & = \begin{bmatrix}
    		1 & 0 & 0 & 0\\
        0 & 0 & 0 & 0\\
        0 & 0 & 0 & 0\\
        0 & 0 & 0 & 1\\
    \end{bmatrix} 
  \end{align*}
  
  Then the possible outcomes are 
  \begin{itemize}
  	\item $+1$
  	\begin{itemize}
  		\item[] For $T$: with probability $\norm{P_{+1}\ket{\psi_{T}}}_2^2 = \norm{\frac{1}{\sqrt{2}}(0, 1, 1, 0)^{T}}_2^2 = 1$, and the sketch is destroyed. 
  		\item[] For $S$: with probability $\norm{P_{+1}\ket{\psi_{S}}}_2^2 = \norm{\frac{1}{2\sqrt{3}}(0, 1, 1, 0)^{T}}_2^2 = 1/6$, and the sketch is destroyed. 
  	\end{itemize}
  	\item $-1$
  	\begin{itemize}
  		\item[] For $T$: with probability $\norm{P_{-1}\ket{\psi_{T}}}_2^2 = \norm{(0, 0, 0, 0)^{T}}_2^2 = 0$, and the sketch is destroyed.
  		\item[] For $S$: with probability $\norm{P_{-1}\ket{\psi_{S}}}_2^2 = \norm{\frac{1}{2\sqrt{3}}(0, 1, -1, 0)^{T}}_2^2 = 1/6$, and the sketch is destroyed.
  	\end{itemize}
  	\item $\bot$
  	\begin{itemize}
  		\item[] For $T$: with probability $\norm{P_{\bot}\ket{\psi_{T}}}_2^2 = \norm{(0, 0, 0, 0)^{T}}_2^2 = 0$. 
 		\item[] For $S$: with probability $\norm{P_{\bot}\ket{\psi_{S}}}_2^2 = \norm{\frac{1}{\sqrt{3}}(1, 0, 0, 1)^{T}}_2^2 = 2/3$, and the updated state is $\frac{1}{\sqrt{2}}(1, 0, 0, 1)^{T} = \ket{\psi_{\set{1, 4}}} = \ket{\psi_{S \setminus \set{2, 3}}}$.
  	\end{itemize}
  \end{itemize}}
  
\noindent
This example is illustrative of an important property of the sketch: the pair
query to sketch $\Qc_T$ always succeeds and returns $+1$ in the example above.
Therefore, when the query is interpreted as asking whether the sketch contains
the pair, the $+1$ always corresponds to \yes{} as desired, since $\{2,3\}
\subseteq T$.

This covers all the operations required of $\Qc_T$, proving
Theorem~\ref{thm:implementability}.  Finally, we will prove Lemma~\ref{lm:reordering}.
\reordering*
\begin{proof}
We proceed by induction on $t$, the number of $\queryone$ and $\querypair$
operations in the sequence. As $\update$ operations are deterministic, cannot
destroy the sketch and do not change the size of the underlying set, the result
holds trivially for $t = 0$. Now, suppose it holds for $t$. The aforementioned
observation about $\update$ means that without loss of generality we may assume
the final update in the sequence is a $\queryone$ or $\querypair$ operation.
Let $T''$ be the underlying set of elements of the sketch before this last
operation is executed, if the sketch is not destroyed by any previous
operation. By the inductive hypothesis, $T''$ is deterministic, and the
probability that no operation before the last one destroys the set is
$\frac{\abs{T''}}{\abs{T}}$.

Now, whether the last operation is $\queryone$ or $\querypair$, let $X$ be the
set of elements queried (so $X$ is either a singleton or a pair). Note that in
each case the sum of the probability of all results that involve destroying the
sketch is $\frac{\abs{X \cap T''}}{\abs{T''}}$. Furthermore, if the sketch is
\emph{not} destroyed, all elements of $X$ are removed from the sketch, and so
$T ' = T'' \setminus X$. Therefore, $T'$ is deterministic as desired. Moreover,
using the inductive hypothesis, the probability that the sketch is not
destroyed after all $t$ query operations is \[
\frac{\abs{T''}}{\abs{T'}}\paren*{1 - \frac{\abs{X \cap T''}}{\abs{T''}}} =
\frac{\abs{T''} - \abs {X \cap T''}}{\abs{T}} = \frac{\abs{T'}}{\abs{T}}
\]
completing the proof.
\end{proof}

\section{Boolean Hidden Matching in the Stream}
\label{sec:bhm}
In the Boolean Hidden Matching problem defined in the streaming setting (see
\cite{GKKRW07}) the algorithm receives a string $x \in \bool^n$ representing set of vertices $[n]$ each labeled by a bit $x_v$, and an $\alpha n$-edge
partial matching $M$ on $\brac{n}$, with labels $z \in \bool^{M}$, as a stream of
elements. The string updates take the form $(v, x_v)$ for $v \in \brac{n}$,
while the matching updates take the form $(uv, z_{uv})$, for $uv \in M$. The
updates can arrive in any order, and the two types of update can be intermingled.

We are promised that either, for every $uv \in M$, $x_u \oplus x_v = z_{uv}$,
or for every such $uv$, $x_u \oplus x_v = \overline{z_{uv}}$. The objective of
the problem is to return $0$ in the first case, and $1$ in the second.

\paragraph{Overview of the algorithm.} 

The algorithm will start with a set $\set{(u, 0, b) \mid u \in [n], b \in \set{0, 1}}$ contained in the sketch. The first element in each tuple represents a vertex, the second represents the current label of this vertex (in the beginning all the labels are set to $0$), while the last label is there for a technical reason described later.

Each time a new element arrives in the stream, our sketch will respond in one
of two ways. If a new vertex label $(v, x_v)$ has arrived and $x_v = 1$,  the
sketch will update the label by performing an $\update$ operation with the swap
$(v, 0, b) \leftrightarrow (v, 1, b)$ for $b \in \set{0, 1}$. If a new labeled
edge $(uv, z_{uv})$ has arrived, the algorithm tries to recover the labels of
$u$ and $v$ from the sketch. To do that, the algorithm makes four $\querypair$
queries:
\begin{align*}
\querypair((u, 0, 0), (v, 0, 0))\\
\querypair((u, 1, 0), (v, 1, 0))\\
\querypair((u, 0, 1), (v, 1, 1))\\
\querypair((u, 1, 1), (v, 0, 1))
\end{align*}
Note that for each of the queries $\querypair((u, a, c), (v, b, c))$, $c$
denotes the parity of the labels of $u, v$ being queried: $c = a \oplus b$.
This guarantees that all the queries in this step are of disjoint pairs, and
are thus independent. Since the edges in the stream form a matching, all the
queries throughout the whole computation are of disjoint pairs as well. 

If one of the queries returns $+1$, the algorithm has recovered $z_{u, v}$ and
a guess at the current labels of $u$ and $v$. Because $\querypair$ is four
times as likely to return $+1$ when both of the elements queried are in the
sketch as when only one is, the guess is correct with at least $2/3$
probability. Then, the algorithm can continue classically by updating the
labels of $u, v$ as updates to them appear, and comparing their parity to
$z_{uv}$ at the end.

This means that, as the probability of a query in each update returning $+1$ is
$\bT{1/n}$, we have 1) a $\bT{\alpha}$ probability of obtaining an answer and
2) a constant-factor gap between the probability of being right and being wrong
in the event we do. So a majority vote of $\bT{1/\alpha}$ copies of the
algorithm run in parallel suffices.

\subsection{Algorithm}

Our algorithm will use the pair sampling sketch with universe $U = [n] \times
\set{0, 1} \times \set{0, 1}$, which will trivially embed in $\brac{2^{\bO{\log
n}}}$. It will be divided into a ``quantum'' stage, which lasts until the
sketch returns something other than $\bot$, and a classical stage, which
processes all the updates after this happens.

We will define the permutation $\pi_v : U \rightarrow U$ as follows: for each
$b \in \set{0,1}$, swap $(v,0,b)$ with $(v,1,b)$.

The algorithm is formally presented as Algorithm \ref{alg:bhm_unordered} below.

\begin{algorithm}
\caption{Quantum Streaming Algorithm for Boolean Hidden Matching}\label{alg:bhm_unordered}
\begin{algorithmic}[1]
\State $\Qc := \create([n]\times \{0\} \times \{0, 1\})$. \\
\textbf{Quantum Stage}
\For{stream updates $\sigma$}
    \If{$\sigma = (w, x_w)$}
        \If{$x_w = 1$}
            \State $\update(\pi_w, \Qc)$
        \EndIf    
    \ElsIf{$\sigma = (uv, z_{uv})$}
        \For{$a, b \in \{0, 1\}$}
            \State $r := \querypair((u, a, a \oplus b), (v, b, a \oplus b), \Qc)$ 
            \If{$r = 1$}
				\State $c \coloneqq a \oplus b \oplus z_{uv}$
				\State Terminate the quantum stage and process all remaining
				updates in the classical stage, sending $uv$ and $c$ to that stage.
				\EndIf
			\If{$r = -1$}
				\State Terminate the algorithm entirely, returning $\bot$.
			\EndIf
        \EndFor 
    \EndIf
\EndFor
\\ \textbf{Classical Stage}
\For{remaining stream updates $\sigma$}
    \If{$\sigma = (w, x_w)$}
        \If{$w = u$ or $w = v$}
            \State $c \coloneqq c \oplus x_w$
        \EndIf
    \EndIf
\EndFor
\If{the algorithm never reached the classical stage}
\State \textbf{return} $\bot$
\Else
\State \textbf{return} $c$
\EndIf
\end{algorithmic}
\end{algorithm}

\subsection{Correctness}

We will now prove the correctness of this algorithm. First, we describe the set
maintained in the sketch.
\begin{lemma}
\label{lm:bhminvariant}
After processing any number of entire updates, if the quantum stage has not yet terminated, the underlying set of $\Qc$ is \[
\set{(v, y_v, b) : v \in \brac{n}, b \in \bool, \text{$(uv,z_{uv})$ has not yet arrived in the stream for any $u \in \brac{n}$}}
\]
where $y_v = x_v$ if the update $(v, x_v)$ has already arrived in the stream, and is $0$ otherwise.
\end{lemma}
\begin{proof}
We proceed by induction on the number of updates processed. When $0$ updates
have been processed, $y = 0^n$ and no updates have yet arrived, so the result
holds by the definition of $\create$.

Now suppose it holds after $t$ updates, and we see a new update $(v, x_v)$.
Then $y_v$ becomes $x_v$. If $x_v = 0$, then no change will be made to the
sketch, and the desired set also has not changed. If $x_v = 1$,
$\update(\pi_v, \Qc)$ is executed, removing $(v,0,b)$ from the sketch and
replacing it with $(v,1,b)$ for each $b \in \bool$, and so the desired set is
still maintained.

Now suppose we see an update $(uv, z_{uv})$. Then four queries will be
performed, and if any return anything other than $\bot$ the quantum stage ends
and so the lemma continues to hold trivially. So suppose they all return
$\bot$. Then $(w,a,b)$ are removed from the sketch for $w \in \set{u,v}$ and
$a,b \in \bool$, which as $u, v$ are now edges incident to an update
$(uv,z_{uv})$ that arrived in the stream, means the desired set is still
maintained.
\end{proof}
Next, we show that the algorithm is more likely to return a correct answer than an
incorrect one.
\begin{lemma}
With probability $\alpha$ the algorithm will return the correct answer, and it
returns an incorrect answer with probability at most $\alpha/2$. Otherwise it
returns $\bot$.
\end{lemma}
\begin{proof}
First note the algorithm can only terminate on updates of the form
$(uv,z_{uv})$. Consider the four $\querypair$ operations performed when
$(uv,z_{uv})$ arrives.  Let $S$ be the underlying set of the sketch before the
operations are performed.  By Lemma~\ref{lm:bhminvariant}, $S$ contains
$(w,y_u,b)$ for $w \in \set{u,v}$ and $b \in \bool$, and no other elements that
can affect the outcome of the query (besides their effect on the size of the
set). 

Therefore, for these four queries $\querypair((u,a,a\oplus b), (v,b,a\oplus
b))$, there is one such that two elements of the pair queried are in $S$, two such that one is, and
one such that neither is. By Lemma~\ref{lm:reordering}, for each of these queries,
the probability that the algorithm does not terminate before making that query
(because of a query in a previous update or an earlier query in this update) is
$1/\abs{S'}$, where $S'$ is the set at query time if the query happens. Note
that the queried pairs are disjoint, so every element queried is in $S'$ iff it
is in $S$. By the definition of $\querypair$, this implies that the probability
that any given query returns $+1$ \emph{and} the algorithm does not terminate before
it is made is \[
\frac{2}{\abs{S'}} \cdot \frac{\abs{S'}}{2n} = \frac{1}{n}
\]
for the query with two elements in $S$, $\frac{1}{4n}$ for each of the queries
with one element in $S$, and $0$ for the other.

Therefore, as there are $\alpha n$ such updates in total, the probability over
all updates that we terminate on a $+1$ returned by a query where both elements
are in $S$ is \[
\frac{1}{n} \cdot \alpha n = \alpha
\]
while the probability we terminate on a $+1$ returned by a query where this is
\emph{not} the case is \[
\frac{1}{4n}\cdot 2 \cdot \alpha n = \frac{\alpha}{2}
\]
from the two out of each batch of four queries that only include one element in $S$.

Therefore, as the algorithm only returns something other than $\bot$ when it
terminates on a query returning $+1$, it will suffice to show that, if the
algorithm terminates on a query where both of the elements are in $S$ at the
time, it always returns the correct answer.

To see this, first note that, by the description of $S$ given earlier in the
proof, that query will be $\querypair((u,y_u,y_u \oplus y_v),(v,y_v,y_u \oplus
y_v))$, performed after seeing the update $(uv,z_{uv})$. So if it terminates
the quantum stage, it is by sending $uv$ and $y_u \oplus y_v \oplus z_{uv}$ to
the classical stage. Now for $w = u,v$, recall that $y_w$ is $x_w$ if $x_w$ has
already appeared in the stream (and therefore will not appear in the classical
stage) and $0$ if it has not (and so will appear in the classical stage).
Moreover, if $x_w$ does appear in the classical stage, it will be XORed onto
the bit we hold. So we are guaranteed to output $x_u \oplus x_v \oplus
z_{uv}$, which is exactly the desired outcome.
\end{proof}
We are now ready to prove correctness.
\begin{theorem}
There exists an quantum algorithm that solves the streaming Boolean Hidden
Matching problem with probability $2/3$ with $\bO{\frac{1}{\alpha}\log{n}}$ quantum memory. The
algorithm is classical except for the use of the quantum sketch described in
Section~\ref{sec:description}.
\end{theorem}
\begin{proof}
By the previous lemma, there is an $\alpha$ probability of returning the
correct answer and no more than an $\alpha/2$ probability of returning an
incorrect one. So by running $\bT{1/\alpha}$ copies of the algorithm in
parallel and taking a majority vote of their outputs, we can return the correct
answer with probability $2/3$.
\end{proof}

\section{Triangle Counting}
\label{sec:triangle}
In this section we prove that the quantum triangle counting algorithm
introduced in~\cite{K21} can be implemented as an algorithm that is classical
other than in using the sketch from Section~\ref{sec:description}.

\begin{definition}[Triangle Counting]
In the (streaming) triangle counting problem, an $m$-edge graph $G = (V,E)$ is
received as a stream of edges in an arbitrary order. The goal of the problem is
to return a number in $((1 - \eps)T,(1 + \eps)T)$ with probability $1 -
\delta$, where $T$ is the number of triangles in $G$: the number of $(u,v,w)
\in V^3$ such that $uv,vw,wu$ are all in $E$, and $\eps,\delta > 0$ are
accuracy parameters.

We are given $T', \Delta_E$ such that $T \ge T'$ and no more than $\Delta_E$
triangles in $G$ share any given edge.
\end{definition}
Formally the space complexity of a triangle counting algorithm is in terms of
$m$, $\Delta_E$, and $T'$, but the standard convention is to assume $T' =
\bOm{T}$ and so we can use $T$ instead.

We will prove the following theorem, a ``black boxed'' version of Theorem~$1$ of~\cite{K21}.
\begin{restatable}{theorem}{trianglecounting}
\label{thm:trianglecounting}
For any $\varepsilon, \delta \in (0,1\rbrack$, there is a streaming
algorithm that uses \[
O\paren*{ \frac{m^{8/5}}{T^{6/5}}\Delta_E^{4/5}\log n
\cdot\frac{1}{\varepsilon^2}\log \frac{1}{\delta}}
\] quantum and classical bits in expectation to return a $(1 \pm
\varepsilon)$-multiplicative approximation to the triangle count in an
insertion-only graph stream with probability $1 - \delta$. This algorithm is
entirely classical except for its use of the sketch in
Section~\ref{sec:description}.

Here $m$ is the number of edges in the stream, $T$ the number of triangles, and
$\Delta_E$ the greatest number of triangles sharing any given edge.
\end{restatable}

The algorithm in~\cite{K21} is based on splitting $T$ into $\Tlk$ and $\Tgk$,
and estimating the former with a quantum algorithm and the latter with a
classical algorithm. We will reproduce the definition of these quantities here.
$k$ will be a parameter that we choose to optimize the trade-off between the
complexity of the two estimators.

\subsection{Triangle Weights}
Fix any ordering of the stream. For any edges $e, f$, we will write $e \bef f$
if $e$ arrives before $f$ in the stream. For any vertices $u,v,w \in V$ such
that $uv, vw \in E$ and $uv \bef vw$, let the degree between $uv$ and $vw$,
$\degb{uvw}$ be the number of edges incident to $v$ that arrive in between $uv$
and $vw$ (not including $uv$ or $vw$ themselves).

For any triple of vertices $(u,v,w) \in V^3$ let \[
\tlk_{uvw} = \begin{cases}
(1 - 1/k)^{\degb{uvw} + \degb{uwv}} &\mbox{if $\set{u, v, w}$ is a triangle in
the graph and $uv \bef uw \bef vw$}\\
0 &\mbox{otherwise.}
\end{cases}
\]
Likewise, let \[
\tgk_{uvw} = \begin{cases}
1 - (1 - 1/k)^{\degb{uvw} + \degb{uwv}} &\mbox{if $\set{u, v, w}$ is a triangle
in the graph and $uv \bef uw \bef vw$}\\
0 &\mbox{otherwise.}
\end{cases}
\]
We will write $\Tlk, \Tgk$ for $\sum_{(u,v,w) \in V^3} \tlk_{uvw}$, $\sum_{(u,v,w)
\in V^3} \tgk_{uvw}$, respectively, so that $T = \Tlk + \Tgk$.

For any vertex $u \in V$, we will write $\Tlk_{u} = \sum_{(v,w) \in V^2}
\tlk_{uvw}$ and $\Tgk_{u} = \sum_{(v,w) \in V^2} \tgk_{uvw}$, so $\sum_{u \in V}
\Tlk_u = \Tlk$ and $\sum_{u \in V} \Tgk_u = \Tgk$.

\subsection{Approach}
By Lemma~11 of~\cite{K21} (reproduced below), $\Tgk$ can be estimated by an
entirely classical algorithm. 
\begin{lemma}
\label{lm:classtimate}
For any $\varepsilon, \delta \in (0,1\rbrack$, there is a classical streaming
algorithm, using \[
O\paren*{\frac{m^{3/2}}{T \sqrt{k}}\Delta_E\log
n\frac{1}{\varepsilon^2}\log \frac{1}{\delta} }
\]
bits of space in expectation, that estimates $\Tgk$ to $\varepsilon T$
precision with probability $1-\delta$.
\end{lemma}
We therefore need to estimate $\Tlk$ with a black-box algorithm, proving:
\begin{lemma}[Black-box version of Lemma~7 of~\cite{K21}]
\label{lm:questimate}
For any $\varepsilon, \delta \in (0,1\rbrack$, there is a streaming
algorithm, using \[
O\paren*{\paren*{\frac{km}{T}}^2\log n \frac{1}{\varepsilon^2}\log
\frac{1}{\delta}}
\]
quantum and classical bits, that estimates $\Tlk$ to $\varepsilon T$
precision with probability $1-\delta$. This algorithm is entirely classical
except for its use of the sketch in Section~\ref{sec:description}.
\end{lemma}

To understand the approach, it is helpful to first consider the case where, for
every triangle $uvw$ with $uv \bef uw \bef vw$, there are no other edges
incident to $v$ or $w$ (note that in this case $T = \Tlk$. In that case, we
could use the following strategy:
\begin{itemize}
\item Start with $2m$ ``dummy'' elements in the sketch.
\item Whenever an edge $vw$ arrives, first query the pairs $((u,v),(u,w))$ for
every $u \in V$. Then swap two dummy elements for $(v,w)$ and $(w,v)$ in the
sketch. Terminate with the value returned by the query if it is anything other
than $\bot$.
\end{itemize}
The query would have a $1/m$ chance\footnote{In fact, the chance grows as the
sketch shrinks from previous queries deleting elements from it. However, by our
Lemma~\ref{lm:reordering}, this is exactly canceled out by the probability of
the algorithm terminating before the query.} of returning $+1$ for every such
triangle $uvw$ in the stream, as when $vw$ arrives, $(u,v)$ and $(u,w)$ would
both be in the sketch. Moreover, every pair query with non-zero expectation
(that is, the ones where both elements of the pair queried are in the sketch at
the time of the query) corresponds to such a triangle: if $(u,v)$ and $(u,w)$
are being queried and both $(u,v)$ and $(u,w)$ are in the sketch, that implies
that $vw$ just arrived and $uv$, $uw$ arrived earlier in the stream. Therefore,
multiplying the output of this algorithm by $m$ gives as an unbiased estimator
for the number of triangles with $\bO{m^2}$ variance, and so averaging
$\bT{\frac{m}{T}}$ copies would suffice for e.g.\ any constant accuracy.

Of course, in practice, other edges can be incident to $v$ and $w$. When do
they cause us problems? When an edge incident to $v$ arrives between $uv$ and
$vw$, or an edge incident to $w$ arrives between $uw$ and $vw$, as the
resulting query will delete $uv$ or $uw$ from the graph. Recalling our
definitions section, this is precisely when $\degb{uvw} + \degb{uwv} > 0$.

Our solution is, for each edge, to only perform the $\querypair$ operation with
probability $1/k$, giving us the chance of avoiding these troublesome edges, at
the cost of only having a $1/k$ chance of querying $((u,v),(u,w))$ when we see
$vw$. So now our probability of returning because of the triangle $uvw$ is
proportional to \[
\frac{1}{k}(1 - 1/k)^{\degb{uvw} + \degb{uwv}} = \frac{\tlk_{uvw}}{k}
\]
if we order $u,v,w$ such that $uv \bef uw \bef vw$, and so we will need to
repeat $\bT{k^2}$ times to estimate $\Tlk$.

One might note here that the definition of $\Tlk$ seems to match what the
algorithm achieves conveniently well. This is of course not a
coincidence:~\cite{K21} splits up $T$ this way precisely because $\Tlk$
captures something the quantum estimator can easily calculate, while $\Tgk$ can
be efficiently calculated by a classical estimator. As the latter point is
already entirely classical, we will not discuss it here: see Section~5
of~\cite{K21} for that algorithm.

\subsection{Algorithm}
We now introduce our estimator for $\Tlk$ using the pair sampling sketch of
Section~\ref{sec:description} as a subroutine. Our universe for the sketch will
be $U = \brac{n}^2  \cup \paren*{[2m] \times \set{S}}$, where $S$ is a fixed
label denoting a ``scratch'' element. Let $g:[m]\to \bool$ be a fully
independent hash function\footnote{Storing such a function would, in general,
require $\bT{m\log m}$ bits. However, we will only query it at any given value
once, so it can be implemented by sampling random bits and discarding them
after use. The use of hash function notation is for readability.} such that \[
g(\ell) = \begin{cases}
1 & \mbox{with probability $1/k$}\\
0 & \mbox{otherwise.}
\end{cases}
\]
\begin{algorithm}
\caption{Quantum streaming algorithm for estimating $\Tlk$}\label{alg:triangle_counting}
\begin{algorithmic}[1]
\State $\Qc := \create([2m] \times \{S\})$. 
\State $\ell := 0$
\For{$uv \in E$:}
	\State $\ell := \ell + 1$
	\If{$g(\ell) = 1$}
		\For{$w \in V$:}
			\State $r := \querypair((w, u), (w, v), \Qc)$
			\If{$r \ne \bot$}
    	    	\textbf{return} $rkm$
    		\EndIf
    	\EndFor
	\EndIf	
	\State $\update(\pi_{\ell, uv}, \Qc)$
\EndFor
\State \textbf{return} 0
\end{algorithmic}
\end{algorithm}

Here $\pi_{\ell, uv}$ denotes the permutation executing the following swaps:
\begin{itemize}
	\item $(2\ell - 1, S)$ for $(u, v)$
	\item $(2\ell, S)$ for $(v, u)$ 
\end{itemize}
So as $\ell$ is incremented at each timestep, this operation swaps two
``scratch elements'' for $(u,v)$ and $(v,u)$.

\begin{lemma}
\label{lm:triangle_stateinv}
For all $\ell = 0, \dots, m$, after the algorithm has processed $\ell$ updates,
if it has not yet terminated the underlying set of $\Qc$ is 
\[
	T_{\ell} = \set{(i, S) : i = 2\ell + 1, \dots, 2m} \cup S_\ell
\]
where 
\[
	S_\ell = \set*{(u, v) : \exists i \in \brac{\ell}, uv = \sigma_i, \underset{j=i+1,\dots,\ell}{\forall} \paren*{g(j) = 0 \vee v \not \in \sigma_j}}
\]
and $\sigma_i$ denotes the $i\nth$ edge in the stream. The size of $T_{\ell}$ is $|T_{\ell}| = 2m - 2\ell + |S_{\ell}|$.
\end{lemma}

\begin{proof}
We proceed by induction on $\ell$. For $\ell = 0$, 
\[
	T_{\ell} = \set{(i, S) \mid i \in [2m]}
\]
and so the result holds. Now, for any $\ell \in \brac{m-1}$, suppose that the result
holds after $\ell$ updates.  Let $xy$ (with $x<y$) be the $(\ell + 1)\nth$
update. 

First consider the effect of the $\querypair$ operations. If $g(\ell + 1) = 0$,
they do not happen, and so the underlying set of $\Qc$ is unchanged. If $g(\ell
+ 1) = 1$, we can assume all of them return $\bot$, as otherwise the algorithm
would terminate. Therefore, their effect on the underlying set of $\Qc$ is to delete \[
   \set{(w, z) \in S_{\ell} : z \in \set{x,y}}
\]
from it. Therefore, whether $g(\ell + 1) = 0$ or $1$, the underlying set of $\Qc$ after the $\querypair$ operations and before the $\update$ operation is \[
\set{(i, S) : i = 2\ell + 1, \dots, 2m} \cup S_\ell'
\]
where \[
S_\ell' = \set*{(u, v) : \exists i \in \brac{\ell}, uv = \sigma_i,
\underset{j=i+1,\dots,\ell+1}{\forall} \paren*{g(j) = 0 \vee v \not \in
\sigma_j}}.
\]
Note that $S_\ell' \cup \set{(x,y),(y,x)} = S_{\ell + 1}$. Next, the algorithm will execute $\update(\pi_{(\ell+1), xy}, \Qc)$, and
the set underlying $\Qc$ will become
\[
	\set{(i, S) : i = 2\ell + 3, \dots, 2m} \cup S_{\ell} \cup \set{(x, y), (y,
	x)} = T_{\ell+1}
\]
completing the proof.
\end{proof}
For the rest of the analysis we will write $X$ for the random variable
corresponding to the output of the algorithm.
\begin{lemma}
\label{lm:triangle_exp}
\[
\Expect{X} = \Tlk
\]
\end{lemma}
\begin{proof} 
We can write $X = \sum_{vw \in E}\sum_{u \in V}X_{u,vw}$, where $X_{u,vw}$ is
the random variable that is $0$ if, when $vw$ is the $\ell\nth$ edge in the
stream, $g(\ell) = 0$ and therefore no query is performed in that update, or
the $\querypair((u,v),(u,w))$ operation returns $\bot$, or if any previous
$\querypair$ operation destroys the sketch, and $rkm$ if it returns $r$.

We will show that for every vertex $u \in V$ and edge $vw \in E$, \[
\Expect{X_{u,vw}} = \tlk_{uvw} + \tlk_{uwv}
\]
which by summing over the $X_{u,vw}$ will give us the lemma.

First, consider the case where $\tlk_{uvw} + \tlk_{uwv}= 0$. By the definition
of $\tlk_{uvw}$, this implies that $uvw$ is not a triangle in $G$, or that it
is but at least one of $uv$ and $uw$ arrived after $vw$ in the stream. In
either case, by Lemma~\ref{lm:triangle_stateinv}, at least one of $(u,v)$ and
$(u,w)$ will not be in the sketch at the time the query occurs, if it occurs at
all. If neither of them is, the query is guaranteed to return $\bot$ and so
$\Expect{X_t} = 0$. If exactly one of them is, then either the query returns
$\bot$ or it is equally likely to return $1$ or $-1$, and so again
$\Expect{X_t} = 0$.

Now suppose $\tlk_{uvw} + \tlk_{uwv} > 0$. Without less of generality, let
$\tlk_{uvw}$ be the one that is non-zero, so $uv \bef uw \bef vw$. Now, if $vw$
is the $\ell\nth$ edge to arrive, this query happens iff $g(\ell) = 1$, and if
so the set underlying the sketch before any of the queries in the $\ell\nth$
update are performed is $T_{\ell - 1}$ by Lemma~\ref{lm:triangle_stateinv}. As
the queries performed in one update are all disjoint, $(u,v)$ and $(u,w)$ are
therefore in the sketch at the time the query is performed iff they are in
$T_\ell$.

By the definition of $T_{\ell-1}$, $uv, uw$ are in the set iff there are no edges incident
to $v$ that arrive after $uv$ in an update $\ell' < \ell$ such that $g(\ell')
= 1$, \emph{and} there are no edges incident to $w$ that arrive after $uw$ in
an update $\ell' < \ell$ such that $g(\ell') = 1$. Over the random choice of
$g$, the probability of this happening is $(1 - 1/k)^{\degb{uvw} + \degb{uwv}}$.

Let $\Fc$ be the event that this happens, and let $\Ec$ be the event that the
query occurs (i.e.\ the algorithm does not terminate before the query takes
place, and $g(\ell) = 1$). Then, by the specification of $\querypair$, and the fact that the
algorithm returns $rkm$ when the query returns $r$, \[
\Expect{X_t | \Ec, \Fc} = \frac{2km}{\abs{S'}}
\]
where $S'$ is the set just before the query occurs, and by
Lemma~\ref{lm:reordering} and the random choice of $g$ (noting that $g$ is fully
independent and so $g(\ell)$ is independent of $\Fc$, $\prob{\Ec} =
\frac{\abs{S'}}{2km}$, so \[
\Expect{X_t | \Fc} = 1.
\]
Clearly $\Expect{X_t | \overline{\Fc}, \overline{\Ec}} = 0$, and if the query
happens but $\Fc$ doesn't, there is at most one of $uv, uw$ in the sketch at
the time, and so by the same argument as in the $\tlk_{uvw} + \tlk_{uwv} = 0$
case, $\Expect{X_t | \overline{\Fc},\Ec} = 0$ too.

So we conclude that \[
\Expect{X_t} = \Expect{X_t | \Fc} \prob{\Fc} = (1 - 1/k)^{\degb{uvw} + \degb{uwv}} = \tlk_{uvw} + \tlk_{uwv}
\]
concluding the proof.
\end{proof}

\begin{lemma}
\label{lm:triangle_var}
\[
	\Var{X} \le (km)^2
\]
\end{lemma}
\begin{proof}
This follows from the fact that $\abs{X} \le km$.
\end{proof}

\begin{lemma}
\label{lm:triangle_questimate}
For any $\varepsilon, \delta \in (0,1\rbrack$, there is a quantum streaming
algorithm, using 

\[
	O\paren*{\paren*{\frac{km}{T}}^2\log n \frac{1}{\varepsilon^2}\log
	\frac{1}{\delta}}
\]

quantum and classical bits, that estimates $\Tlk$ to $\varepsilon T$
precision with probability $1-\delta$.
\end{lemma}

\begin{proof}
By running the algorithm $\bT{\paren*{km/T\varepsilon}^2}$ times in parallel
and averaging the outputs, we obtain an estimator with expectation $\Tlk$ and
variance at most $\varepsilon^2 T^2/3$. So by Chebyshev's inequality the
estimator will be within $\varepsilon T$ of $\Tlk$ with probability $2/3$. We
can then repeat this $\bT{\log \frac{1}{\delta}}$ times and take the median of
our estimators to estimate $\Tlk$ to $T\varepsilon$ precision with probability
$1 - \delta$.

Each copy of the algorithm uses $\bO{\log{n}}$ qubits since the size of the
universe is $|U| = \bO{n^2}$ and $\bO{\log n}$ classical bits, and so the
result follows.
\end{proof}

We can now prove the main theorem of this section.
\trianglecounting*
\begin{proof}
By combining Lemma~\ref{lm:triangle_questimate} with Lemma~\ref{lm:classtimate}
(Lemma~7 of~\cite{K21}) and setting $k = \frac{T^{2/5}}{m^{1/5}}
\Delta_E^{2/5}$.
\end{proof}

\section{Maximum Directed Cut}
\label{sec:dicut}
In this section, we prove that the quantum streaming algorithm for Maximum
Directed Cut (max-dicut) introduced in~\cite{KPV24} can be implemented as an
algorithm that is classical other than in its use of our quantum sketch from
Section~\ref{sec:description}.
\begin{definition}[Max-Dicut]
In the (streaming) maximum directed cut problem with approximation factor
$\alpha$, a directed graph $G = (V,E)$ is received one edge at a time in the
stream. The objective of the algorithm is to return an $\alpha$-approximation to \[
\mdicut(G) = \max_{x \in \bool^{V}} \sum_{\dedge{uv} \in E} (1 - x_u)x_v
\]
defined as returning a value in $\brac{\alpha\cdot \mdicut(G), \mdicut(G)}$.
\end{definition}
Equivalently, $\mdicut(G)$ is the maximum, over all partitions $V = V_0 \sqcup
V_1$, of the number of edges $\dedge{uv} \in E$ such that $u \in V_0$ and $v
\in V_1$.

In \cite{KPV24}, the authors designed a quantum streaming algorithm that
$0.4844$-approximates the value of max-dicut in polylogarithmic quantum space.
This algorithm, following the classical results of~\cite{SSSV23b,S23}, makes
use of the fact that an approximation to $\mdicut(G)$ can be derived from a
``snapshot'' of $G$ that groups the vertices $v \in V$ by their biases $b_v =
\frac{\dout_v - \din_v}{d_v}$, where $\dout_v$ is the number of edges with $v$
as their head, $\din_v$ is the number of edges with $v$ as their tail, and $d_v
= \dout_v + \din_v$.
\begin{definition}[Snapshot]
Let $\tb \in \brac{-1,1}^\ell$ be a vector of bias thresholds. The (first-order)
snapshot $\hist{G} \in \Nbb^{\ell \times \ell}$ of $G = (V,E)$ is given by: \[
\hist{G}_{i,j} = \abs{\set{\dedge{uv} \in E}: u \in H_i, v \in H_j}
\]
where $H_i$ is the $i\nth$ ``bias class'', given by \[
H_i = \begin{cases}
\set{v \in V : b_v \in \interval{\tb_i,\tb_{i+1}}} & \mbox{$i \in \brac{\ell - 1}$}\\
\set{v \in V : b_v \in \brac{\tb_{\ell}, 1}} & \mbox{$i = \ell$.}
\end{cases}
\]
\end{definition}
The approximation algorithms need only a constant number of bias classes, so
$\ell = \bO{1}$.

More precisely, the algorithm of~\cite{KPV24} approximates an object called the
``pseudosnapshot'', a version of the snapshot with some noise added that adds
some errors to the biases of the vertices, to account for the fact that a small
change in the bias of a high-degree vertex could substantially change the
counts in the snapshot if it happened to be near a border. They then show that
using this instead of the snapshot does not introduce too much error.

The quantum algorithms component of~\cite{KPV24} is entirely encapsulated in
their Lemma~10, which establishes that a quantum algorithm can estimate this
pseudosnapshot with adequate accuracy. In our Lemma~\ref{lm:psnapest}, we give
a version of this lemma where the algorithm is entirely classical other than in
using our quantum sketch. As both the definition of the pseudosnapshot and
Lemma~\ref{lm:psnapest} are quite ornate, we will defer them until later in
this section. As this lemma is identical to Lemma~10 of~\cite{KPV24} outside of
the ``black-box'' use of the quantum sketch, combining it with the other
results in~\cite{KPV24} immediately gives our result.
\begin{theorem}
\label{thm:mdcutalg}
There is a streaming algorithm which $0.4844$-approximates the \mdcut{} value
of an input graph $G$ with probability $1 - \delta$. The algorithm uses
$\bO{\log^5 n \log \frac{1}{\delta}}$ qubits of space. This algorithm is
classical other than in using the quantum sketch from
Section~\ref{sec:description}.
\end{theorem}

To illustrate our techniques, we start by considering a simpler problem, the
``heavy edges'' problem. In this problem, the goal is to estimate the number of
directed edges $\dedge{uv}$ in the stream such that $u$ and $v$ have high
degree amongst the edges that arrive before $\dedge{uv}$. This problem can be
solved in $\bO{\log{n}}$ by a classical algorithm with access to the ``pair
counting'' sketch, and might itself be of interest as a primitive for other
applications. The algorithm itself captures the key idea of the
``pseudosnapshot'' algorithm while avoiding some of the complications that make
pseudosnapshot estimation difficult.

\subsection{Counting Heavy Edges}
\label{sec:heavy_edges}
In this section we show how to use the quantum sketch to solve the following
problem: given a stream of directed edges, how many of them are incident to a
vertex that has \emph{already} seen the arrival of a large number of edges
incident to it?

Let $d_v^{\leq \dedge{uv}}$ be the number of edges incident to $v$ that arrive
before $\dedge{uv}$ in the stream, including $\dedge{uv}$ itself.

\begin{definition}[Heavy Edges Problem]
In the $(d_H, d_T)$-heavy edges problem, with accuracy parameter $\eps$, we
receive a stream of directed edges and our goal is to estimate, to $\pm \eps m$
accuracy, the number of edges $\dedge{uv}$ in the stream such that: \[
d_u^{\leq\dedge{uv}} \geq d_H \text{ and } d_v^{\leq \dedge{uv}} \geq d_T \]
\end{definition}

\paragraph{Intuition}
For each vertex $u$ the algorithm will store prefixes (under the order given by
$i$) of $\set{(u, i, H) \mid i \in [\min\{d_H, d_u\}]}$ and $\set{(u, i, T) \mid
i \in [\min\{d_T, d_u\}]}$. The label $\set{H, T}$ describes whether this element
is needed to check the degree of the ``head'' of the directed edge, or the
``tail''. The counter allows to keep track of the degree of the vertex $u$.

The algorithm starts with a sketch containing only ``scratch'' states that do
not depend on the graph. For each edge $\dedge{uv}$ arriving in the stream, the
algorithm takes two actions:
\begin{itemize}
\item Update the information about the endpoints' degrees $u, v$ by shifting
the ``stacks'' of elements by 1 using a permutation $(u, i, h) \leftrightarrow
(u, (i+1), h)$, and swap 4 scratch states into the bottom of the ``stacks'',
$(\ell, s) \to (u, 0, h)$, and the same for $v$ and $t$. 
\item Check whether $\dedge{uv}$ is a heavy edge by making a query
$\querypair((u, d_H, h), (v, d_T, t))$. If the edge is heavy, both of the
elements will be present and so the expected value of the query conditioned on
not returning $\bot$ will be positive. Otherwise at most one will be and the
expected value will be $0$.
\end{itemize}
Note that if the query returns $\bot$, causing the elements queried to be
deleted from the sketch, the elements $(u, d_H-1, h), (v, d_T-1, t)$ are still
present in the sketch, and will become $(u, d_H, h)$ and $(v, d_T, t)$ the next
time an edge neighboring $u$ or $v$ is encountered. 

By Lemma~\ref{lm:reordering}, for each $\dedge{uv}$, the sketch survives until
the query to $((u, d_H, h), (v, d_T, t))$ with probability inversely
proportional to the probability of the query not returning $\bot$, and so the
expected value of the algorithm is proportional to the number of edge
$\dedge{uv}$ that are ``heavy''.

\subsubsection{Algorithm}

We will use an instance of the pairs sketch from Section~\ref{sec:description},
with universe $U$ containing $[n] \times [2n] \times \{H, T\} \cup [4m] \times
\{S\}$ where $S, H, T$ are fixed labels, denoting ``scratch'', ``head'', and
``tail'', respectively (this can be trivially embedded in $\brac{n^{\bO{1}}}$).

\begin{algorithm}
\caption{Quantum streaming algorithm for $(d_H, d_T)-$heavy edges problem}\label{alg:heavy_edges}
\begin{algorithmic}[1]
\State $\Qc := \create([4m] \times \{S\})$. 
\State $\ell := 0$
\For{$\dedge{uv}$ in the stream:}
	\State $\update(\pi_{\ell, uv}, \Qc)$
	\State $\ell := \ell + 4$
	\State $r := \querypair((u, d_H, H), (v, d_T, T), \Qc)$
	\If{$r \ne \bot$}
        \State Return $2rm$ and stop the computation \label{step:heavy_return}
    \EndIf
\EndFor
\State Return 0
\end{algorithmic}
\end{algorithm}

Here $\pi_{\ell, uv}$ denotes the permutation that does the following:
\begin{itemize}
	\item $\forall i \in [2n-1], w \in \{u, v\}, Q \in \{H, T\}: (w, i, Q)
	\rightarrow (w, (i+1) \bmod d_Q, Q)$
	\item Swap:
	\begin{itemize}
	\item $(\ell, S)$ for $(u, 0, H)$
	\item $(\ell+1, S)$ for $(u, 0, T)$
	\item $(\ell+2, S)$ for $(v, 0, H)$
	\item $(\ell+3, S)$ for $(v, 0, T)$ 
	\end{itemize}
\end{itemize}
As this consists of disjoint cycles and swaps, it is a permutation.

\subsubsection{Correctness}

Here we state the main steps of the correctness proof of this algorithm. The proofs of the lemmas can be found in the appendix \ref{appendix:heavy_edges}.

In the beginning of the stream, our quantum sketch $\Qc$ contains the $S_0 =
[4m] \times \{S\}$, $|S_0| = 4m$. We will use $S_{m'}$ to denote the set stored
in the sketch after processing $m'$ edges (that is, $m'$ iterations of the loop
in the algorithm). We will now prove a loop invariant on the contents of the
sketch in the case where it has not yet terminated. We will write $d_v^{<m'}$
for $d_v^{<\dedge{xy}}$, for $\dedge{xy}$ the ${m'}\nth$ edge to arrive in
the stream.

\begin{restatable}{lemma}{hedgeloopinvariant}
\label{lm:hedgeloopinvariant}
	If the algorithm has not yet terminated after processing $m'$ edges,
	\begin{align*}
	S_{m'} =& \bigcup_{w \in V}
	\paren*{\set{(w, j, H)\mid j \in [\min(d_H, d_w^{<m'}) - 1]} \cup \set{(w, j, T)\mid j \in
	[\min(d_T, d_w^{<m'})-1]}}\\ &\cup \{(i, S) \mid 4m' \leq i \leq 4m\}
	\end{align*}
\end{restatable}

By the properties of the sketch and the previous lemma, we get the following statements about the expectation and the varience of the output.

\begin{restatable}{lemma}{hedgeexpect}
The expectation of the estimator returned by the algorithm is the number of
edges $\dedge{uv}$ such that $d_u^{\leq \dedge{uv}} \geq d_H,
d_v^{\leq \dedge{uv}} \geq d_T$.
\end{restatable}

\begin{restatable}{lemma}{hedgevar}
    The variance of the estimator returned by the algorithm is $\bO{m^2}$. 
\end{restatable}

And finally, by repeating the same algorithm many times in parallel and taking the average, we obtain the final result.

\begin{restatable}{lemma}{hedgethm}
There exists an quantum algorithm that approximates the number of edges
$\dedge{uv}$ with $d_u^{\leq \dedge{uv}} \geq d_H$ and $d_u^{\leq \dedge{uv}}
\geq d_H$ to $O(\eps m)$ error with probability $2/3$ and using $O(\frac{1}{\eps^2}
\log{n})$ quantum memory. This algorithm is classical except for the use of the quantum
sketch described in Section~\ref{sec:description}. 
\end{restatable}

\subsection{Pseudosnapshot Definition}
Now we proceed to the pseudosnapshot algorithm. We will first define the object
we need to approximate, from Section~5 of~\cite{KPV24}.

Let $(d_i)_{i=0}^{\floor{\log_{1+\eps^3}n}}$ be given by $d_i = \floor{(1 +
\eps^3)^i}$ for $i < \floor{\log_{1+\eps^3}  n}$ and $d_{\floor{\log_{1 +
\eps^3}  n}} = n$.  Let $\accuracy \le \poly{n}$, $\eps \in \brac{0,1}$ be
accuracy parameters to be chosen later, and let
$(f_i)_{i=0}^{\floor{\log_{1+\eps^3}  n}}$ be a family of fully independent
random hash functions such that $f_i : E \rightarrow \bool$ is $1$ with
probability $\accuracy/2d_i$, while $g : V \rightarrow \brac{-\eps,\eps}$ is a
fully independent random hash function that is uniform on
$\brac{-\eps,\eps}$.\footnote{We adopt this notation for the sake of clarity,
but we will only ever evaluate $f_i(e)$ at the update when edge $e$ arrives,
and $g$ after processing the entire stream on the endpoints of edges our
algorithm stores, so we do not need to pay the prohibitive overhead of storing
these hash functions. Moreover, while we write $g(e)$ as a random real number,
it will only ever be used in sums and comparisons with numbers of
$\poly(n,\eps)$ precision, so we do not need to store it any more precision
than that.}

Fix an arrival order for the edges $e$ of the directed graph. For any vertex
$v$ and edge $e$, let $\doutb{e}_v$, $\db{e}_v$, refer to the out-degree and
degree of $v$ when only $e$ and edges that arrive before $e$ are counted, and
let $\douta{e}_v$, $\da{e}_v$ refer to these quantities when counting only
edges that arrive \emph{after} $e$. Let $\wt{i}$ be the largest $i$ such that
$d_i < \db{e}$. Then define $\dbps{e}_v = d_{\wt{i}}$, and let $\doutbps{e}_v$
be the number of edges $e'$ with head $v$ that arrive before $e$ and have
$f_{\wt{i}}(e') = 1$, multiplied by $2d_{\wt{i}}/\accuracy$.  We will then define the
$e$-pseudobias of $v$, $\bps{e}_v$, as \[ \min\set*{2\frac{\doutbps{e}_v +
\douta{e}_v}{\dbps{e}_v + \da{e}_v} - 1 + g(v),
1}\text{.}
\]
In other words, the $\bps{e}_v$ is the bias of $v$ when its degree among $e$ and
edges that arrive before $e$ is rounded to the bottom of the interval
$\interval{d_{\wt{i}},d_{\wt{i}+1}}$, and its out-degree among these edges is
estimated using the number of out-edges ``sampled'' by $f_{\wt{i}}$, with a
small amount of noise $g(v)$ added. Since this can sometimes produce a
pseudobias larger than 1, we then cap it at 1.

\begin{definition}
Let $\tb \in \brac{-1,1}^\ell$ be a vector of bias thresholds. The
pseudosnapshot $\histps{G} \in \Nbb^{\ell \times \ell}$ of $G = (V,E)$ is given
by: \[
\histps{G}_{i,j} = \abs{\set{\dedge{uv} \in E: u \in H^{\dedge{uv}}_i, v \in
H^{\dedge{uv}}_j}}
\]
where $H^e_i$ is the $i\nth$ ``$e$-pseudobias'' class, given by \[
H^e_i = \begin{cases}
\set{v \in V : \bps{e}_v \in \interval{\tb_i,\tb_{i+1}}} & \mbox{$i \in \brac{\ell - 1}$}\\
\set{v \in V : \bps{e}_v \in \brac{\tb_{\ell}, 1}} & \mbox{$i = \ell$.}
\end{cases}
\]
The restriction of $\histps{G}$ to $E' \subseteq E$ is then given by: \[
\histps{G,E'}_{i,j} =  \abs{\set{\dedge{uv} \in E': u \in H^{\dedge{uv}}_i, v \in
H^{\dedge{uv}}_j}}
\]
\end{definition}

\subsection{Pseudosnapshot Estimation Algorithm}
In this section we give a small-space algorithm for estimating the
pseudosnapshot of a directed graph $G = (V,E)$, restricted to edges $e =
\dedge{uv}$ with $\db{e}_u \in \interval{d_i,d_{i+1}}$, $\db{e}_v \in
\interval{d_j, d_{j+1}}$ for some specified $i, j \in
\brac{\floor{\log_{1+\eps^3} n}} \cup \set{0}$. The algorithm will be entirely
classical other than its use of the sketch from Section~\ref{sec:description},
and will satisfy exactly the requirements of the algorithm described in
Lemma~10 of~\cite{KPV24}.

\paragraph{Intuition} As in the heavy edges algorithm described in Section
\ref{sec:heavy_edges}, for each vertex $u$ the algorithm stores ``stacks'' of
elements $\set{(u, i, Q): i \in I }$ in the sketch for index sets $I$ and labels $Q$, and
queries specific $(u,i,Q)$ to determine whether the stack has ``reached''
i.

In the ``heavy edges'' algorithm the counter kept track of the degree of the
vertex $u$. Specifically, it was keeping track of whether the degree of $u$ has
reached the threshold by storing the set $\{(u, i, H)\mid i \in [\min\{d_H,
d_u\}]\}$, and a similar set for $d_T$ and $T$. We need to keep track of both
the in- and out-degree of $u$, as we are interested in whether its bias lies in
a certain range. Unfortunately, it is not possible to maintain two counters per
vertex, as we can only query \emph{pairs}, and so checking whether $u$ and $v$
both are present with in- and out-degree in the right range would not be
possible\footnote{This might raise the question: why can we check if $u$ and
$v$ have degrees in a \emph{range} at all, as that requires checking two
predicates per vertex? However, we can avoid this by e.g.\ counting how often
they have degree at least $d$ and separately counting how often they have
degree at least $d'$, and taking the difference.}. Instead, we track the
degree, and then encode information about the out-degree in the higher-order
bits of the degree, by increasing the degree by $k$ with probability $1/k$
whenever an out-edge is seen, for some $k$ larger than the largest degree we
will allow $v$ to have. Then our queries will look at the stack at multiple
locations $d_H, d_H + k, d_H + 2k$, etc.

\begin{restatable}{lemma}{psnapest}
\label{lm:psnapest}
Let $\alpha, \beta \in \brac{\floor{\log_{1+\eps^3}  n}} \cup \set{0}$. Let
$\accuracy = \poly(n)$ be the integer accuracy parameter used in defining the
pseudobiases. Fix a draw of the hash functions $(f_i)_{i =
0}^{\ceil{\log_{1+\eps^3} n}}$, and therefore the pseudobiases of the graph
$G$.

Then there is a classical algorithm that uses the ``pair sampling'' sketch that, if $\sum_{e \in E}
(f_\alpha(e) + f_\beta(e)) \le 2\accuracy m $, returns an estimate of the
pseudosnapshot of $G$ restricted to edges $e = \dedge{uv}$ with $\db{e}_u \in
\interval{d_\alpha,d_{\alpha+1}}$, $\db{e}_v \in \interval{d_\beta,
d_{\beta+1}}$.

Each entry of the estimate has bias at most the number of edges $\dedge{uv}$ such that:
\begin{enumerate}
\item $\db{\dedge{uv}}_u \in \interval{d_\alpha,
d_{\alpha+1}}$
\item $\db{\dedge{uv}}_v \in \interval{d_\beta, d_{\beta+1}}$
\item $\max\set*{\frac{\accuracy}{2d_{\alpha}} \doutbps{\dedge{uv}}_u + 1,
\frac{\accuracy}{2d_{\beta}} \doutbps{\dedge{uv}}_v + 1} > \accuracy$
\end{enumerate}

Each entry has variance $\bO{\accuracy^3m^2}$. The algorithm uses $\bO{\log n}$
qubits of space. 
\end{restatable}
\noindent
The algorithm will maintain an instance of the quantum sketch from
Section~\ref{sec:description} with underlying set \[
\Sc \cup \bigcup_{i=1}^{2\accuracy^2} (\Ac^i \cup \Bc^i \cup \Cc^i \cup \Dc^i)
\] of size $M' \le M = C\accuracy^3m$ for some sufficiently large constant $C$. $M'$ will
start at $M$ and decrease as queries remove elements from the set. 
All the sets are disjoint and are subsets of the following universe. \[
U = \set{(\ell, S) \mid \ell \in [M]} \cup \set{(u, j, Q) \mid u \in [n], j \in [Mn]}
\] where the last element $Q \in \{a^{i}\mid i \in [2\accuracy^2]\} \cup
\{b^{i}\mid i \in [2\accuracy^2]\} \cup \{c^{i}\mid i \in [2\accuracy^2]\} \cup
\{d^{i}\mid i \in [2\accuracy^2]\}$ or $S$ is a label indicating which of
$\Ac^i, \Bc^i, \Cc^i, \Dc^i$, or $\Sc$ a tuple belongs to.

$\Sc$ will contain our ``scratch elements''. It is equal to \[
\set{(j, S) :  j = \ell, \dots, M} 
\]
where $\ell$ is incremented every time we use a scratch element
(by swapping it with some other element we want), and the label $S$ 
indicates that this is a scratch element. The algorithm will keep track of $\ell$ 
so that it knows which element to swap from.

    The remaining sets encode information about vertices and their degrees. For
$\Ec = \Ac, \Bc, \Cc, \Dc$, each takes the form
\begin{align*}
\Ec^i &= \bigcup_{v \in V}\Ec_v^i\\
\Ec_v^i &= \set{(v,j,e^i)\mid j \in E_v}
\end{align*}
where the sets $E_v = A_v, B_v, C_v, D_v$ are not explicitly stored in the sketch
but directly correspond to the elements contained in it, and $e = a, b,
c, d$ marks which of $\Ac^i, \Bc^i, \Cc^i, \Dc^i$ an element belongs to, for
$i \in \brac{2\accuracy^2}$. Note that the underlying set $E_v$ for each
$\Ec_v^i$ does not depend on $i$---the $2\accuracy^2$ copies are to allow us to
perform a larger set of queries, and the sets will remain identical
until the algorithm terminates.  The four non-scratch components are broken
into two pairs (as each vertex can be both a head and a tail, and the two cases
need to be handled separately), with each pair having an element for checking
whether degrees are high enough, and one for checking whether degrees are low
enough (both will track whether vertices have had enough edges coming out of
them that pass the hash functions).

\begin{itemize}
\item $\Ac^i$ are for tracking vertices with degree at least $d_\alpha$.
\item $\Bc^i$ are for tracking when those vertices have their degree exceed
$d_{\alpha+1}$.
\item $\Cc^i$ are for tracking vertices with degree at least $d_\beta$.
\item $\Dc^i$ are for tracking when those vertices have their degree exceed
$d_{\beta+1}$.
\end{itemize}
  
At the start of the execution of our algorithm, $A_v, B_v, C_v, D_v =
\emptyset$. They will be updated with the following three operations:

\begin{itemize}
\item $\inc(\Ec, v, r)$ replaces $E_v$ with $\set{i + r : i \in E_v} \cup
\brac{r}$, and replaces $\ell$ with $\ell + 2\accuracy^2r$. Note that this can be accomplished
with a single $\update$ operation to the sketch with a permutation that sends $(v,i,e^j)$ to $(v,(i+r) \bmod Mn,e^j)$
for each $i,j$, and then swapping the first $2\accuracy^2r$ remaining elements of $\Sc$ with
$\cup_{i=1}^r (v,i,e^j)$ for each $j \in \brac{2\accuracy^2}$. 
\end{itemize}

\begin{itemize}
\item $\measure(u,v)$ makes the following $\querypair$ queries to the sketch $\Qc$ for each $(i,j) \in \brac{\accuracy}^2$:

\begin{align*}
  r^{1}_{i, j} = \querypair&((u, (d_\alpha + (i-1)d_{\alpha+1}), a^{t_{i,j}}),&(v, (d_\beta + (j-1)d_{\beta+1}), c^{t_{i,j}}), &\Qc)
\\r^{2}_{i, j} =\querypair&((u, (id_{\alpha+1})               , b^{t_{i,j}}),&(v, (d_\beta + (j-1)d_{\beta+1}), c^{s_{i,j}}), &\Qc)
\\r^{3}_{i, j} =\querypair&((u, (d_\alpha + (i-1)d_{\alpha+1}), a^{s_{i,j}}),&(v, (jd_{\beta+1})              , d^{t_{i,j}}), &\Qc)
\\r^{4}_{i, j} =\querypair&((u, (id_{\alpha+1})               , b^{s_{i,j}}),&(v, (jd_{\beta+1})              , d^{s_{i,j}}), &\Qc)
\end{align*}

Where the $s_{i,j}$, $t_{i,j}$ are chosen so that these pairs are all
disjoint to one another across all $i,j$ (note that this is possible
because we have $\accuracy^2$ possible values to choose from). If the result of
any of these queries is anything other than $\bot$, the
quantum part of the algorithm will terminate and the remaining execution will
be entirely classical.
\end{itemize}
    
\begin{itemize}
\item $\cleanup(u,v)$ makes the following $\queryone$ queries for all $w \in \set{u, v}$, $i \in \brac{M}$, and $j \in \brac{2\accuracy^2}$
\begin{align*}
  \queryone&((w,(d_\alpha + id_{\alpha+1}),a^{j}))
\\\queryone&((w,((i+1)d_{\alpha+1}),b^{j}))
\\\queryone&((w,(d_\beta + id_{\beta+1}),c^{j}))
\\\queryone&((w,((i+1)d_{\beta+1}),d^{j}))
\end{align*}

If anything other than
the $\bot$ is returned in any of these queries, the algorithm halts entirely and
outputs a zero estimate for the pseudosnapshot.

Together, the effect of performing $\measure(u,v)$ and $\cleanup(u,v)$, if they
do \emph{not} return something other than$\bot$, is to
delete the following elements from $A_w$, $B_w$, $C_w$, $D_w$, for $w = u,v$
and for all $i \in \brac{M}$ (note that these elements may have not been
present to begin with, or may be ``removed'' multiple times between the two
operations---this does not cause any issues):
\begin{itemize}
\item $d_\alpha + (i-1)d_{\alpha + 1}$ from $A_w$.
\item $id_{\alpha + 1}$ from $B_w$.
\item $d_\beta + (i-1)d_{\alpha + 1}$ from $C_w$.
\item $id_{\beta + 1}$ from $D_w$.
\end{itemize}
\end{itemize}
We can now describe the algorithm.

\paragraph{Initialization} Create the quantum sketch with only scratch elements: 
$$\Qc := \create([M]\times \set{S})$$

\paragraph{Quantum Stage} For each $\dedge{uv}$ processed until the quantum
stage terminates:
\begin{enumerate}
\item $\inc(\Ec,w,1)$ for $w = u,v$, and $\Ec = \Ac, \Bc, \Cc, \Dc$.
\item If $f_\alpha(\dedge{uv}) = 1$, $\inc(\Ac,u,d_{\alpha+1})$ and
$\inc(\Bc,u,d_{\alpha+1})$.
\item If $f_\beta(\dedge{uv}) = 1$, $\inc(\Cc,u,d_{\beta + 1})$ and
$\inc(\Dc,u,d_{\beta+1})$.
\item $\measure(u,v)$. If the result of the query $r^{x}_{i, j}$ is something other than $\bot$, pass $r^{x}_{i, j}$ 
along with $x, i, j, u, v$ to the classical stage and continue.
\item $\cleanup(u,v)$. If any of the queries returns anything other than the
$\bot$, immediately terminate the
algorithm, outputting an all-zeroes estimate.
\end{enumerate}
If the quantum stage processes every edge without being terminated by a
query outcome, output an all-zeroes estimate and skip the classical
stage.
    
\paragraph{Classical Stage} 
This stage is reached once the $\measure$ operation terminates the quantum
stage, with $r^{x}_{i, j}$ passes along with $x, i, j, u, v$.

For the remainder of the stream, track $\douta{e}_u$,$\douta{e}_v$,
$\da{e}_u$,$\da{e}_v$ (giving us exact values for these variables). Then
estimate $\db{e}_u$, $\db{e}_v$ by assuming that they are equal to $d_\alpha$,
$d_\beta$, respectively. Then estimate $\doutbps{e}_u$, $\doutbps{e}_v$, by
assuming that the number of edges $e$ with head $u$ and $f_\alpha(e) = 1$ is
$i-1$, and the number with head $v$ and $f_\beta(e) = 1$ is $j-1$.

Combine these estimates and evaluate $g(u)$, $g(v)$ to estimate $\bps{e}_u$ and
$\bps{e}_v$. 

If $x = 1$ or $4$, set $r^{x}_{i, j} M/2$ as the corresponding entry of
the pseudosnapshot estimate (with every other entry as $0$). If $x = 2$ or $3$,
set it as $- r^{x}_{i, j} M/2$ instead. 

\subsection{Analysis}
\subsubsection{Sketch Invariant}

\begin{lemma}
\label{lm:stinv}
Consider any time after some number of edges have been (completely) processed
in the quantum stage, and suppose the stage has not yet terminated, and $t$ has
not exceeded $M$. Let $v \in V$, and let $r$ be the number of those edges that
were incident to $v$. Let $R$ be the number of those edges $e$ such that $v$
was the head of the edge, and $f_\alpha(e) = 1$. 

Then, if $R = 0$, 
\begin{align*}
A_v &= \Nbb \cap \interval{1,\min(r+1,d_\alpha)}\\
B_v &= \Nbb \cap \interval{1,\min(r+1,d_{\alpha +1})})
\end{align*}
and if $R > 0$, there exists $(\prefix_i)_{i=1}^R \in \brac{r}^R$ such that:
\begin{align*}
A_v &= \Nbb \cap \paren*{\interval{1,d_\alpha} \cup \bigcup_{i=1}^R I_i}\\
B_v &= \Nbb \cap \paren*{\interval{1,d_{\alpha+1}} \cup \bigcup_{i=1}^R J_i}
\end{align*}
Where
\begin{align*}
I_i &= \begin{cases}
\interval{d_\alpha+(i-1)d_{\alpha+1}+\prefix_i,d_\alpha + id_{\alpha + 1}} & \mbox{i
< R}\\
\interval{d_\alpha + (R-1)d_{\alpha + 1} + \prefix_R, Rd_{\alpha + 1}
+ \min(r+1, d_{\alpha})} & \mbox{i = R}
\end{cases}\\
J_i &= \begin{cases}
\interval{id_{\alpha+1}+\prefix_i,(i+1)d_{\alpha + 1}} & \mbox{i
< R}\\
\interval{Rd_{\alpha + 1} + \prefix_R, Rd_{\alpha + 1}
+ \min(r+1, d_{\alpha+1})} & \mbox{i = R}
\end{cases}
\end{align*}
The same relationship holds for $C_v$ and $D_v$, except with $\beta$ instead of $\alpha$.
\end{lemma}

\begin{proof}
We will prove the result for $A_v$ and $B_v$. The proof for $C_v$ 
and $D_v$ is identical, with $\beta$ substituted for $\alpha$. Note that when
describing the updates performed on seeing an edge $\dedge{uv}$ or $\dedge{vu}$
we will ignore the updates that only touch $u$, as they have no effect on the
sets we are analyzing here.

We prove this result by induction. First, we consider $R = 0$ and $r = 0$ as
the base case. Next, we prove the inductive step where $r$ is increased by 1
while $R = 0$. This finishes the proof for $R = 0$ and any $r$. Next, we prove
the result for $R = 1$ by considering the step when $R$ is increased from $R =
0$. Note that any update that increases $R$ also increases $r$, so in this case
$r$ is also increased by 1. Last, we fix any $r$ and any $R > 0$ and prove the
inductive step when $r$ is increased by 1, but $R$ stays unchanged, and when
both $r$ and $R$ are increased by 1, which completes the proof.

We start with the case where $R = 0$ and $r =
0$. Then the statement is equivalent to $A_v = B_v = \emptyset$, which follows from
how we have defined the initial state of the algorithm, and the fact that edges
not incident to $v$ do not result in updates that affect $A_v$ or $B_v$.

Now suppose the result holds for $r$ (with $R$ still $0$), and consider the $(r
+ 1)\nth$ edge incident to $v$ to arrive (with no $\dedge{vw}$ such that
$f(\dedge{vw}) = 1$ having arrived yet). Then the update consists of performing
$\inc(\Ac,v,1)$, $\inc(\Bc,v,1)$, and then $\measure(v, w)$ or $\measure(w,v)$
for some $w$, followed by $\cleanup(v,w)$ or $\cleanup(w,v)$. 

Before the
$\measure$ operation, this gives us
\begin{align*}
A_v &= \Nbb \cap \interval{1,\min(r+2,d_\alpha + 1)}\\
B_v &= \Nbb \cap \interval{1,\min(r+2,d_{\alpha + 1} + 1)})
\end{align*}
and as the condition of the lemma supposes that the queries did not result
in the algorithm terminating, it will have removed $d_\alpha$ from $A_v$ and
$d_{\alpha+1}$ from $B_v$, and so 
\begin{align*}
A_v &= \Nbb \cap \interval{1,\min(r+2,d_\alpha)}\\
B_v &= \Nbb \cap \interval{1,\min(r+2,d_{\alpha + 1})})
\end{align*}
as desired.
    
Next we will prove the result for $R$ increased from $0$ to $1$ and $r$ increased by $1$. By the previous section, when $R = 0$ the state is
\begin{align*}
A_v &= \Nbb \cap \interval{1,\min(r+1,d_\alpha)}\\
B_v &= \Nbb \cap \interval{1,\min(r+1,d_{\alpha + 1})})\text{.}
\end{align*}
When the $(r+1)\nth$ edge incident to $v$ arrives, if it is also the first edge
with head $v$ in $f_{\alpha}^{-1}(1)$, the update consists of performing
$\inc(\Ac,v,1)$, $\inc(\Bc,v,1)$, $\inc(\Ac,v,d_{\alpha + 1})$,
$\inc(\Bc,v,d_{\alpha + 1})$, and then $\measure(v, w)$ for
some $w$, followed by $\cleanup(v,w)$. 

    After the $\inc$ operations are performed we have 
\begin{align*}
A_v &= \Nbb \cap \interval{1,\min(r+d_{\alpha + 1} + 2,d_\alpha + d_{\alpha +
1} + 1)}\\
B_v &= \Nbb \cap \interval{1,\min(r+d_{\alpha+1} + 2,2d_{\alpha + 1} +
1)})
\end{align*}
and then after performing the queries we remove $d_\alpha$ and $d_\alpha +
d_{\alpha+1}$ from $A_v$ and $d_{\alpha + 1}$ and $2d_{\alpha + 1}$ from $B_v$,
so we have
\begin{align*}
A_v &= \Nbb \cap (\interval{1,d_\alpha} \cup \interval{d_\alpha + 1,
\min(r+d_{\alpha + 1} + 2,d_\alpha + d_{\alpha + 1})})\\
B_v &= \Nbb \cap (\interval{1,d_{\alpha+1}} \cup \interval{d_{\alpha+1} + 1,
\min(r+d_{\alpha+1} + 2,2d_{\alpha + 1})}))
\end{align*}
which matches the lemma statement by setting $\prefix_1 = \prefix_R = 1$.

We now need to consider two more cases to complete the proof: First, the case when we have the result for
some $R > 0$ and $r$, and the update consists of the $(r+1)\nth$ edge incident
to $v$, and this edge is \emph{not} in $f_{\alpha}^{-1}(1)$ or does not have
$v$ as its head.  Second, the case when we have the result for $R$ and $r$ and
the update consists of the $(r+1)\nth$ edge incident to $v$, which is also the
$(R+1)\nth$ edge with head $v$ in $f_{\alpha}^{-1}(1)$.

Let us deal with the cases when $r$ is increased from $r$ to $r+1$ and $R$
stays the same first. The actions performed during the update are
$\inc(\Ac,v,1)$, $\inc(\Bc,v,1)$, and then $\measure(v, w)$ or
$\measure(w,v)$ for some $w$, followed by $\cleanup(v,w)$ or
$\cleanup(w,v)$. The sets before the update, by the inductive hypothesis, are
\begin{align*}
A_v &= \Nbb \cap \paren*{\interval{1,d_\alpha} \cup \bigcup_{i=1}^R I_i}\\
B_v &= \Nbb \cap \paren*{\interval{1,d_{\alpha+1}} \cup \bigcup_{i=1}^R J_i}
\end{align*}
with the intervals $I_i$ and $J_i$ as defined in the lemma statement. The
$\inc(\Ac, v, 1)$ and $\inc(\Bc, v, 1)$ operations, then, correspond to
replacing these intervals with
\begin{align*}
I_i' &= \begin{cases}
\interval{d_\alpha+(i-1)d_{\alpha+1}+\prefix_i + 1,d_\alpha + id_{\alpha + 1} + 1} & \mbox{i
< R}\\
\interval{d_\alpha + (R-1)d_{\alpha + 1} + \prefix_R + 1, Rd_{\alpha + 1}
+ \min(r+2, d_{\alpha} + 1) } & \mbox{i = R}
\end{cases}\\
J_i' &= \begin{cases}
\interval{id_{\alpha+1}+\prefix_i + 1,(i+1)d_{\alpha + 1} + 1} & \mbox{i
< R}\\
\interval{Rd_{\alpha + 1} + \prefix_R + 1, Rd_{\alpha + 1}
+ \min(r+2, d_{\alpha+1} + 1)} & \mbox{i = R}
\end{cases}
\end{align*}
and $\interval{1,d_\alpha}$ with $\interval{1,d_\alpha + 1}$, and
$\interval{1,d_{\alpha+1}}$ with $\interval{1,d_{\alpha+1}+1}$.

The queries then delete $d_\alpha + (i-1)d_\alpha$ from $A_v$ and
$id_\alpha$ from $B_v$ for every $i \in [M]$. So $\interval{1,d_\alpha + 1}$ and
$\interval{1,d_{\alpha+1}}$ return to $\interval{1,d_\alpha}$,
$\interval{1,d_{\alpha+1}}$, respectively, and the intervals $I_i$, $J_i$
become
\begin{align*}
I_i'' &= \begin{cases}
\interval{d_\alpha+(i-1)d_{\alpha+1}+\prefix_i + 1,d_\alpha + id_{\alpha + 1}} & \mbox{i
< R}\\
\interval{d_\alpha + (R-1)d_{\alpha + 1} + \prefix_R + 1, Rd_{\alpha + 1}
+ \min(r+2, d_{\alpha}) } & \mbox{i = R}
\end{cases}\\
J_i'' &= \begin{cases}
\interval{id_{\alpha+1}+\prefix_i + 1,(i+1)d_{\alpha}} & \mbox{i
< R}\\
\interval{Rd_{\alpha + 1} + \prefix_R + 1, Rd_{\alpha + 1}
+ \min(r+2, d_{\alpha+1})} & \mbox{i = R}
\end{cases}
\end{align*}
and so by setting the new $I_i$ to be $I_i''$ and likewise with $J_i$, the
sets are in the form desired (by incrementing every element of
$(\prefix_i)_{i=1}^R$ by 1), as $r$ is now $1$ larger.

    Finally, we consider the inductive step where both $R$ and $r$ are increased by
$1$. This means the arriving edge is the $(r+1)\nth$ edge incident to $v$,
which is also the $(R+1)\nth$ edge with head $v$ in $f^{-1}_\alpha(1)$.

 The update consists of performing $\inc(\Ac,v,1)$, $\inc(\Bc,v,1)$,
 $\inc(\Ac,v,d_{\alpha + 1})$, $\inc(\Bc,v,d_{\alpha + 1})$, and then
 $\measure(v, w)$ for some $w$, followed by $\cleanup(v,w)$. The sets before
 the update, by the inductive hypothesis, is
\begin{align*}
A_v &= \Nbb \cap \paren*{\interval{1,d_\alpha} \cup \bigcup_{i=1}^R I_i}\\
B_v &= \Nbb \cap \paren*{\interval{1,d_{\alpha+1}} \cup \bigcup_{i=1}^R J_i}
\end{align*}
with the intervals $I_i$ and $J_i$ as defined in the lemma statement.

As $\inc$
operations add together, after all the increment operations we have replaced
$\interval{1,d_\alpha}$ with $\interval{1,d_\alpha + d_{\alpha + 1} + 1}$ and
$\interval{1,d_{\alpha+1}}$ with $\interval{1,2d_{\alpha+1}  + 1}$ and $I_i$,
$J_i$ with 
\begin{align*}
I_i' &= \begin{cases}
\interval{d_\alpha+id_{\alpha+1}+\prefix_i + 1,d_\alpha + (i+1)d_{\alpha + 1} + 1}
& \mbox{i < R}\\
\interval{d_\alpha + Rd_{\alpha + 1} + \prefix_R + 1, (R+1)d_{\alpha + 1}
+ \min(r+2, d_{\alpha} + 1) } & \mbox{i = R}
\end{cases}\\
J_i' &= \begin{cases}
\interval{(i+1)d_{\alpha+1}+\prefix_i + 1,(i+2)d_{\alpha + 1} + 1} & \mbox{i
< R}\\
\interval{(R+1)d_{\alpha + 1} + \prefix_R + 1, (R+1)d_{\alpha + 1}
+ \min(r+2, d_{\alpha+1} + 1)} & \mbox{i = R}
\end{cases}
\end{align*}
and so after the queries, as they remove $d_{\alpha} + (i-1)d_{\alpha +
1}$ from $A_v$ and $id_{\alpha + 1}$ from $B_v$ for all $i$,

$\interval{1,d_\alpha + d_{\alpha + 1} + 1}$ becomes $\interval{1,d_\alpha}
\cup \interval{d_\alpha + 1, d_\alpha + d_{\alpha + 1}}$ and
$\interval{1,2d_{\alpha+1}  + 1}$ becomes $\interval{1,d_{\alpha+1}} \cup
\interval{d_{\alpha + 1}  + 1, 2d_{\alpha +1}}$. $I_i'$ and $J_i'$ become
\begin{align*}
I_i'' &= \begin{cases}
\interval{d_\alpha+id_{\alpha+1}+\prefix_i + 1,d_\alpha + (i+1)d_{\alpha + 1}}
& \mbox{i < R}\\
\interval{d_\alpha + Rd_{\alpha + 1} + \prefix_R + 1, (R+1)d_{\alpha + 1}
+ \min(r+2, d_{\alpha}) } & \mbox{i = R}
\end{cases}\\
J_i'' &= \begin{cases}
\interval{(i+1)d_{\alpha+1}+\prefix_i + 1,(i+2)d_{\alpha + 1}} & \mbox{i
< R}\\
\interval{(R+1)d_{\alpha + 1} + \prefix_R + 1, (R+1)d_{\alpha + 1}
+ \min(r+2, d_{\alpha+1})} & \mbox{i = R}
\end{cases}
\end{align*}
and so the sketch is of the form required, by setting (writing $\prefix_i'$ for the
old $\prefix_i$, and noting that this means $\prefix_i \le r+1$ for all $i$)

\begin{align*}
I_i &= \begin{cases}
\interval{d_\alpha + 1, d_\alpha + d_{\alpha + 1}} & \mbox{$i = 1$}\\
I_{i-1}'' &\mbox{$1 < i \le R+1$}
\end{cases}\\
J_i &= \begin{cases}
\interval{d_{\alpha + 1}  + 1, 2d_{\alpha +1}} & \mbox{$i = 1$}\\
J_{i-1}'' &\mbox{$1 < i \le R+1$}
\end{cases}\\
\prefix_i &= \begin{cases}
1 & \mbox{i = 1}\\
\prefix_{i - 1}' + 1 & \mbox{$1 < i \le R+1$}
\end{cases}
\end{align*}
which completes the proof.

\end{proof}

\subsubsection{Query Outcomes}
In this section we characterize the expected contribution to the pseudosnapshot
estimate from all of the query outcomes that can terminate the quantum
stage of the algorithm: those that return something other than $\bot$ in $\measure$
and $\cleanup$.

\begin{lemma}
For all $u,v$, the expected contribution of $\cleanup(u,v)$ to every entry of
the pseudosnapshot estimate is $0$.
\end{lemma}
\begin{proof}
This follows immediately from the fact that we do not modify the estimate when
the algorithm terminates due to $\cleanup$.
\end{proof}

\begin{lemma}
\label{lm:rightdegreecontribution}
For every $\dedge{uv} \in E$ such that $\db{\dedge{uv}}_u \in
\interval{d_{\alpha}, d_{\alpha+1}}$, $\db{\dedge{uv}}_v \in
\interval{d_{\beta}, d_{\beta+1}}$, and for every $(i,j) \in
\brac{\accuracy}^2$, the total expected contribution of $\measure(u,v)$ to the
pseudosnapshot estimate from returning one of
$r^{x}_{i,j} {x\in\brac{4} \neq \bot}$ is $1$ to the
$\bps{\dedge{uv}}_u$, $\bps{\dedge{uv}}_v$ entry and zero to all other entries
if $i = \frac{\accuracy}{2d_{\alpha}} \doutbps{\dedge{uv}}_u + 1$ and $j =
\frac{\accuracy}{2d_{\beta}} \doutbps{\dedge{uv}}_v + 1$.  Otherwise it is zero
everywhere.
\end{lemma}

\begin{proof}

We will start by analyzing the expectation conditional on the algorithm not
terminating before $\dedge{uv}$ arrives. This will give us the result in terms
of $M'$. We will then use Lemma~\ref{lm:reordering} to give us the expectation
in terms of $M$.

At the time when $\dedge{uv}$ arrives, let $R_u'$, $R_v'$, be the $R$ in
Lemma~\ref{lm:stinv} for $u,v$ respectively, and likewise for $r_u'$, $r_v'$
and $r$. Then, for $w = u,v$, set $r_w = r_w' + 1$. Set $R_u = R_u' +
f_\alpha(\dedge{uv}) + f_\beta(\dedge{uv})$. Note that this means that, before
making queries, the algorithm will update the sketch with an $\inc(\Ec, w, 1)$
operation for $w = u,v$ and $\Ec = \Ac, \Bc, \Cc, \Dc$, 
an $\inc(\Ec,u,d_{\alpha+1}\cdot f_{\alpha}(\dedge{uv}))$ operation for $\Ec = \Ac, \Bc$
and an
$\inc(\Ec,u,d_{\beta+1}\cdot f_{\beta}(\dedge{uv}))$ operation for $\Ec = \Cc, \Dc$.

Note that $R_u = \frac{\accuracy}{2d_{\alpha}}\doutbps{\dedge{uv}}_u$, $R_v =
\frac{\accuracy}{2d_{\beta}} \doutbps{\dedge{uv}}_v$, $r_u = \db{\dedge{uv}}_u$, $r_v =
\db{\dedge{uv}}_v$. So we have $r_u \in \interval{d_\alpha,d_{\alpha+1}}$, $r_v
\in \interval{d_\beta,d_{\beta+1}}$.

By Lemma~\ref{lm:stinv}, this means that for all $i$, before the increments,
\begin{align*}
d_\alpha + id_{\alpha+1} - 1 \in A_u, id_{\alpha+1} - 1 \in B_u &\mbox{ iff $i
\in \brac{R_u'}$}\\
d_\beta + id_{\beta+1} - 1 \in C_v, id_{\beta+1} - 1 \in D_v &\mbox{ iff $i \in
\brac{R_v'}$}
\end{align*}
and after them:
\begin{align*}
d_\alpha + id_{\alpha+1} \in A_u, id_{\alpha+1} \in B_u &\mbox{ iff $i
\in \brac{R_u}$}\\
d_\beta + id_{\beta+1} \in C_v, id_{\beta+1} \in D_v &\mbox{ iff $i \in
\brac{R_v}$}
\end{align*}

Now, we write $\Qc$ for the state of the sketch before the
queries (where the underlying set is of size $M'$), and recall that
\begin{align*}
  r^{1}_{i, j} = \querypair&((u, (d_\alpha + (i-1)d_{\alpha+1}), a^{t_{i,j}}),& (v, (d_\beta + (j-1)d_{\beta+1}), c^{t_{i,j}}), &\Qc)
\\r^{2}_{i, j} =\querypair&((u, (id_{\alpha+1})               , b^{t_{i,j}}),& (v, (d_\beta + (j-1)d_{\beta+1}), c^{s_{i,j}}), &\Qc)
\\r^{3}_{i, j} =\querypair&((u, (d_\alpha + (i-1)d_{\alpha+1}), a^{s_{i,j}}),& (v, (jd_{\beta+1})              , d^{t_{i,j}}), &\Qc)
\\r^{4}_{i, j} =\querypair&((u, (id_{\alpha+1})               , b^{s_{i,j}}),& (v, (jd_{\beta+1})              , d^{s_{i,j}}), &\Qc)
\end{align*}

we have
\begin{align*}
\Pr[r^{1}_{i, j} = +1] &= \begin{cases}
\frac{2}{M'}  & \mbox{$i \in \brac{R_u + 1}$ and $j \in \brac{R_v + 1}$}\\
\frac{1}{2M'} & \mbox{exactly one of $i \in \brac{R_u + 1}$ and $j \in \brac{R_v + 1}$}\\
0 & \mbox{otherwise}
\end{cases}\\
\Pr[r^{1}_{i, j} = -1] &= \begin{cases}
0  & \mbox{$i \in \brac{R_u + 1}$ and $j \in \brac{R_v + 1}$}\\
\frac{1}{2M'} & \mbox{exactly one of $i \in \brac{R_u + 1}$ and $j \in \brac{R_v + 1}$}\\
0 & \mbox{otherwise}
\end{cases}\\
\Pr[r^{2}_{i, j} = +1] &= \begin{cases}
\frac{2}{M'}  & \mbox{$i \in \brac{R_u}$ and $j \in \brac{R_v + 1}$}\\
\frac{1}{2M'} & \mbox{exactly one of $i \in \brac{R_u}$ and $j \in \brac{R_v + 1}$}\\
0 & \mbox{otherwise}
\end{cases}\\
\Pr[r^{2}_{i, j} = -1] &= \begin{cases}
0  & \mbox{$i \in \brac{R_u}$ and $j \in \brac{R_v + 1}$}\\
\frac{1}{2M'} & \mbox{exactly one of $i \in \brac{R_u}$ and $j \in \brac{R_v + 1}$}\\
0 & \mbox{otherwise}
\end{cases}\\
\Pr[r^{3}_{i, j} = +1] &= \begin{cases}
\frac{2}{M'}  & \mbox{$i \in \brac{R_u + 1}$ and $j \in \brac{R_v}$}\\
\frac{1}{2M'} & \mbox{exactly one of $i \in \brac{R_u + 1}$ and $j \in \brac{R_v}$}\\
0 & \mbox{otherwise}
\end{cases}\\
\Pr[r^{3}_{i, j} = -1] &= \begin{cases}
0  & \mbox{$i \in \brac{R_u + 1}$ and $j \in \brac{R_v}$}\\
\frac{1}{2M'} & \mbox{exactly one of $i \in \brac{R_u + 1}$ and $j \in \brac{R_v}$}\\
0 & \mbox{otherwise}
\end{cases}\\
\Pr[r^{4}_{i, j} = +1] &= \begin{cases}
\frac{2}{M'}  & \mbox{$i \in \brac{R_u}$ and $j \in \brac{R_v}$}\\
\frac{1}{2M'} & \mbox{exactly one of $i \in \brac{R_u}$ and $j \in \brac{R_v}$}\\
0 & \mbox{otherwise.}
\end{cases}\\
\Pr[r^{4}_{i, j} = -1] &= \begin{cases}
0  & \mbox{$i \in \brac{R_u}$ and $j \in \brac{R_v}$}\\
\frac{1}{2M'} & \mbox{exactly one of $i \in \brac{R_u}$ and $j \in \brac{R_v}$}\\
0 & \mbox{otherwise.}
\end{cases}\\
\end{align*}

    Now, recall that the contribution to the chosen entry of the pseudosnapshot
(with the choice depending only on $i,j$) when seeing the result
$r^{a}_{i, j}$  is $r^{a}_{i, j}M/2$ if $a = 1,4$ and $-r^{a}_{i, j}M/2$
if $a = 2,3$. Therefore, writing the total expected contribution from queries
summed over possible outcomes $r^{a}_{i, j} = \pm1$ as $x_{i,j}^a$, we have

\begin{align*}
x_{i,j}^1 &= \begin{cases}
\frac{M}{M'} & \mbox{$i \in \brac{R_u + 1}$ and $j \in \brac{R_v + 1}$}\\
0 & \mbox{otherwise}
\end{cases}\\
x_{i,j}^2 &= \begin{cases}
-\frac{M}{M'} & \mbox{$i \in \brac{R_u}$ and $j \in \brac{R_v + 1}$}\\
0 & \mbox{otherwise}
\end{cases}\\
x_{i,j}^3 &= \begin{cases}
-\frac{M}{M'} & \mbox{$i \in \brac{R_u + 1}$ and $j \in \brac{R_v}$}\\
0 & \mbox{otherwise}
\end{cases}\\
x_{i,j}^4&= \begin{cases}
\frac{M}{M'} & \mbox{$i \in \brac{R_u}$ and $j \in \brac{R_v}$}\\
0 & \mbox{otherwise}
\end{cases}
\end{align*}
and so \[
\sum_{a=1}^4x_{i,j}^a = \begin{cases}
\frac{M}{M'} & \mbox{$i = R_u + 1$ and $j = R_v + 1$}\\
0 & \mbox{otherwise.}
\end{cases}
\]

Now, when $i = R_u + 1$ and $j = R_v + 1$, and $\db{\dedge{uv}}_u \in
\interval{d_{\alpha}, d_{\alpha+1}}$, $\db{\dedge{uv}}_v \in
\interval{d_{\beta}, d_{\beta+1}}$, the entry of the pseudosnapshot estimate
chosen is $\bps{\dedge{uv}}_u, \bps{\dedge{uv}}_v$. We have that the
contribution to it conditioned on the algorithm not terminating before
processing $\dedge{uv}$ is $M/M'$. So, as the contribution from $\dedge{uv}$ is
guaranteed to be $0$ if the algorithm has already terminated, the total
expectation is, by Lemma~\ref{lm:reordering}, \[ 
(1 - p) \cdot \frac{M}{M'} = (1 - p) \cdot \frac{M}{(1-p)M} = 1
\]
where $p$ is the probability of termination before processing $\dedge{uv}$.

\end{proof}

\begin{lemma}
\label{lm:wrongdegreecontribution}
For every $\dedge{uv}$ such that $\db{\dedge{uv}}_u \not\in
\interval{d_{\alpha}, d_{\alpha+1}}$, or $\db{\dedge{uv}}_v \not\in
\interval{d_{\beta}, d_{\beta+1}}$, the total expected contribution of
$\measure(u,v)$ to any entry of the pseudosnapshot estimate is zero.
\end{lemma}

\begin{proof}
We will analyze the expectation conditional on the algorithm not
terminating before $\dedge{uv}$ arrives. As the expected contribution is
trivially $0$ conditional on the algorithm terminating before $\dedge{uv}$
arrives, this will suffice for the result. 

We have that at least one of $\db{\dedge{uv}}_w < d_\gamma$ for some
$(w,\gamma) \in \set{(u,\alpha), (v,\beta)}$ or $\db{\dedge{uv}}_w \ge
d_{\gamma + 1}$ for some $(w,\gamma) \in \set{(u,\alpha), (v,\beta)}$. By
Lemma~\ref{lm:stinv}, this implies that, for one of $(E,F) \in \set{(A,B),
(C,D)}$, \[
d_\gamma + (i - 1)d_{\gamma+1} \in E_w \Leftrightarrow id_{\gamma+1} \in F_w
\]
for all $i \in \brac{M}$. 

So, writing $r^{a}_{i, j}$ for the outcomes of queries in $\measure(u, v)$ for the state
of the algorithm, we have either 
\begin{align*}
\Pr[r^{1}_{i, j} = +1] &= \Pr[r^{2}_{i, j} = +1],& \Pr[r^{1}_{i, j} = -1] &= \Pr[r^{2}_{i, j} = -1]\\
\Pr[r^{3}_{i, j} = +1] &= \Pr[r^{4}_{i, j} = +1],& \Pr[r^{3}_{i, j} = -1] &= \Pr[r^{4}_{i, j} = -1]\\
\end{align*}

for all $i,j$, or 
\begin{align*}
\Pr[r^{1}_{i, j} = +1] &= \Pr[r^{2}_{i, j} = +1],& \Pr[r^{3}_{i, j} = -1] &= \Pr[r^{3}_{i, j} = -1]\\
\Pr[r^{2}_{i, j} = +1] &= \Pr[r^{4}_{i, j} = +1],& \Pr[r^{2}_{i, j} = -1] &= \Pr[r^{4}_{i, j} = -1]\\
\end{align*}
for all $i,j$. 

    As for all $i,j$ the contribution from the query
result $r^{a}_{i, j}$ is made to the same entry of the estimate
regardless of $a$, and is $r^{a}_{i, j}M/2$ for $a = 1,4$ and $-r^{a}_{i, j}M/2$
for $a = 2,3$, this implies the expected contributions cancel out, and so the
lemma follows.

\end{proof}

\begin{lemma}
\label{lm:bias}
Each entry of the estimate has bias at most the number of edges
$\dedge{uv}$ such that:
\begin{enumerate}
\item $\db{\dedge{uv}}_u \in \interval{d_\alpha,
d_{\alpha+1}}$
\item $\db{\dedge{uv}}_v \in \interval{d_\beta, d_{\beta+1}}$
\item $\max\set*{\frac{\accuracy}{2d_{\alpha}} \doutbps{\dedge{uv}}_u + 1,
\frac{\accuracy}{2d_{\beta}} \doutbps{\dedge{uv}}_v + 1} > \accuracy$
\end{enumerate}
\end{lemma}

\begin{proof}
By Lemma~\ref{lm:wrongdegreecontribution}, no edge edge with endpoints in the
wrong degree classes contribute to the estimate. By
Lemma~\ref{lm:rightdegreecontribution} every edge $\dedge{uv}$ with endpoints
in the right degree classes contributes $1$ to the correct entry of the
estimate of $\histps{G}$ (and 0 to all others), unless \[
\max\set*{\frac{\accuracy}{2d_{\alpha}} \doutbps{\dedge{uv}}_u + 1,
\frac{\accuracy}{2d_{\beta}} \doutbps{\dedge{uv}}_v + 1} > \accuracy
\]
in which case it contributes 0 to all entries, as the $i,j$ in
Lemma~\ref{lm:rightdegreecontribution} only range in $\brac{\accuracy}^2$.
\end{proof}

\begin{lemma}
The variance of any entry of the pseudosnapshot estimate is $\bO{\accuracy^6m^2}$.
\end{lemma}
\begin{proof}
The output of the algorithm is an estimate of the pseudosnapshot with one
$\bO{M}$ entry and all other entries $0$. So this follows by the fact that $M =
\bO{\accuracy^3 m}$.
\end{proof}

\subsubsection{Space Usage}
\begin{lemma}
The algorithm uses only $\bO{\log n}$ qubits of space.
\end{lemma}
\begin{proof}
The algorithm maintains a quantum sketch with elements from the universe 
$$U = \{(\ell, S) \mid \ell \in [M]\} \cup \{(u, j, q) \mid u \in [n], j \in [Mn]\}$$
$q \in \{a^{i}\mid i \in [2\accuracy^2]\} \cup \{b^{i}\mid i \in [2\accuracy^2]\} \cup \{c^{i}\mid i \in [2\accuracy^2]\} \cup \{d^{i}\mid i \in [2\accuracy^2]\}$. It requires $\bO{\log{|U|}} = \bO{\log
n}$ qubits to store, along with a constant number of counters of size
$\poly(n)$, along with some rational numbers made from constant-degree
polynomials of such numbers.

\end{proof}

\section*{Acknowledgements}

Nadezhda Voronova thanks Robin Kothari for his question at QIP 2024, which led to further work and is addressed in this paper. She was supported by a Sloan Research Fellowship to Mark Bun.

John Kallaugher and Ojas Parekh were supported by the Laboratory Directed
Research and Development program at Sandia National Laboratories, a
multimission laboratory managed and operated by National Technology and
Engineering Solutions of Sandia, LLC., a wholly owned subsidiary of Honeywell
International, Inc., for the U.S. Department of Energy's National Nuclear
Security Administration under contract DE-NA-0003525. Also supported by the
U.S. Department of Energy, Office of Science, Office of Advanced Scientific
Computing Research, Accelerated Research in Quantum Computing program,
Fundamental Algorithmic Research for Quantum Computing (FAR-QC). 

Sandia National Laboratories is a multi-mission laboratory managed and operated
by National Technology \& Engineering Solutions of Sandia, LLC (NTESS), a wholly
owned subsidiary of Honeywell International Inc., for the U.S. Department of
Energy’s National Nuclear Security Administration (DOE/NNSA) under contract
DE-NA0003525. This written work is authored by an employee of NTESS. The
employee, not NTESS, owns the right, title and interest in and to the written
work and is responsible for its contents. Any subjective views or opinions that
might be expressed in the written work do not necessarily represent the views
of the U.S. Government. The publisher acknowledges that the U.S. Government
retains a non-exclusive, paid-up, irrevocable, world-wide license to publish or
reproduce the published form of this written work or allow others to do so, for
U.S. Government purposes. The DOE will provide public access to results of
federally sponsored research in accordance with the DOE Public Access Plan.

\newpage
\bibliographystyle{alpha}
\bibliography{refs}

\newpage

\appendix
\section{Omitted Proofs from Section \ref{sec:heavy_edges}}
\label{appendix:heavy_edges}
In this section, we present the omitted proofs from
Section~\ref{sec:heavy_edges}.

In the beginning of the stream, our quantum sketch $\Qc$ contains the $S_0 =
[4m] \times \{S\}$, $|S_0| = 4m$. We will use $S_{m'}$ to denote the set stored
in the sketch after processing $m'$ edges (that is, $m'$ iterations of the loop
in the algorithm). We will now prove a loop invariant on the contents of the
sketch in the case where it has not yet terminated. We will write $d_v^{<m'}$
for $d_v^{<\dedge{xy}}$, for $\dedge{xy}$ the ${m'}\nth$ edge to arrive in
the stream.

\hedgeloopinvariant*
\begin{proof}
	We will proceed by induction. Before processing any edges, $S_0 = \{(i, S)
	\mid 0 \leq i \leq 4m\}$ by the definition of the $\create$ operation.

	Now suppose
	$$S_{m'} = \{(i, S) \mid 4m' \leq i \leq 4m\} \cup \bigcup_{w \in V}
	(\bigcup_{j \in [\min(d_H, d_w) - 1]} (w, j, H) \cup  \bigcup_{j \in
	[\min(d_T, d_w)-1]} (w, j, T))$$
	and the $(m'+1)$-th edge is $\dedge{uv}$.

	First, the algorithm applies the $\update(\pi_{(4m'), uv}, \Qc)$ operation.
	By the definition of the permutation, and using the fact that $d_w^{<m'}$
	is $d_w^{<m'+1} + 1$ for $w = u,v$ and $d_w^{<m'}$ otherwise (and likewise,
	the permutation only touches $(w,j,Q)$ for $w = u,v$), this replaces the
	set with
	\begin{align*}
	&\{i, S) \mid 4(m'+1) \leq i \leq 4m\} \\
  &\cup \bigcup_{w \in V\setminus\{u,
	v\}} \paren*{\bigcup_{j \in [\min(d_H, d_w^{<m'+1}) - 1]} (w, j, H) \cup
	\bigcup_{j \in [\min(d_T, d_w^{<m'+1})-1]} (w, j, T)}\\ &\cup
	\paren*{\bigcup_{j \in [\min(d_H+1, d_w^{<m'+1})-1]} (w, j, H) \cup
	\bigcup_{j \in [\min(d_T+1, d_w^{<m'+1})-1]} (w, j, T)}
	\end{align*}
	Next, the algorithm executes $\querypair((u, d_H, H), (v, d_T, T), \Qc)$.
	If this query returns anything other than $\bot$, the algorithm terminates
	before $m'$ edges are processed and so the lemma holds trivially.
	Otherwise, the only effect of the query is the deletion of elements $(u,
	d_H, H)$ and $(v, d_T, T)$, if they are present in the set. Therefore 
	\begin{align*}
	S_{m'+1} &= \{(i, S) \mid 4(m'+1) \leq i \leq 4m\} \\
  &\phantom{=}\cup \bigcup_{w \in
	V\setminus\{u, v\}} \paren*{\bigcup_{j \in [\min(d_H, d_w) - 1]} (w, j, H) \cup
	\bigcup_{j \in [\min(d_T, d_w)-1]} (w, j, T)}\\ 
	&\phantom{=}\cup \paren*{\bigcup_{j \in
	[\min(d_H, d_w)-1]} (w, j, H) \cup  \bigcup_{j \in [\min(d_T, d_w)-1]} (w,
	j, T)} \\
	&= \{(i, S) \mid 4(m'+1) \leq i \leq 4m\} \cup \bigcup_{w \in V}
	\paren*{\bigcup_{j \in [\min(d_H, d_w) - 1]} (w, j, H) \cup \bigcup_{j \in
	[\min(d_T, d_w)-1]} (w, j, T)}
	\end{align*}
\end{proof}
Using this invariant, we show that the queries deliver the correct result in
expectation.
\hedgeexpect*
\begin{proof}
We will show that, for each edge $\dedge{uv}$, the expectation of that edge's
contribution to the estimator (that is, the probability that the algorithm
terminates while processing that edge, times the expectation of the value the
algorithm returns if this happens) is $1$ if $d_u^{\leq \dedge{uv}} \geq d_H,
d_v^{\dedge{uv}}$ and $0$ otherwise.

Let $\dedge{uv}$ be the ${m'}\nth$ such edge.  Let $s_{m'}$ be the size of
$S_{m'}$, given that the algorithm has not terminated when this edge is
processed.  By Lemma~\ref{lm:reordering}, this probability is
$\frac{s_{m'}}{4m}$. By Lemma~\ref{lm:hedgeloopinvariant}, 
  $$S_{m'} = \{(i, S) \mid 4m' \leq i \leq 4m\} \cup \bigcup_{w \in V}
  \paren*{\bigcup_{j \in [\min(d_H, d_w^{<m'}) - 1]} (w, j, H) \cup  \bigcup_{j
  \in [\min(d_T, d_w^{<m'})-1]} (w, j, T)}$$
if the algorithm has not stopped the computation before reaching $\dedge{uv}$. After executing $\update$, the set of stored elements becomes
\begin{align*} 
  &\{(i, S) \mid 4(m'+1) \leq i \leq 4m\} \\
  &\cup \bigcup_{w \in V\setminus\{u,
  v\}} \paren*{\bigcup_{j \in [\min(d_H, d_w^{<m'+1}) - 1]} (w, j, H) \cup
  \bigcup_{j \in [\min(d_T, d_w^{<m'+1})-1]} (w, j, T)}\\ &\cup
  \paren*{\bigcup_{j \in [\min(d_H+1, d_w^{<m'+1})-1]} (w, j, H) \cup
  \bigcup_{j \in [\min(d_T+1, d_w^{<m'+1})-1]} (w, j, T)} 
\end{align*}
Note that the size of this set is $s_{m'-1}$. We can now consider the expected
contribution of $\dedge{uv}$ to the estimator based on three possible cases:
\begin{itemize}
	\item $d_u^{\leq \dedge{uv}} \geq d_H, d_v^{\leq \dedge{uv}} \geq d_T$.
	Then $\querypair$ returns $1$ with probability $2/s_{m'-1}$, resulting in
	the algorithm terminating with output $2m$, and $\bot$
	otherwise. Therefore, the expected contribution is \[
		\frac{s_{m'-1}}{4m} \cdot \frac{2}{s_{m'-1}} \cdot 2m = 1
	\]
	\item Exactly one of $d_u^{\leq \dedge{uv}} \geq d_H$ or $d_v^{\leq
	\dedge{uv}} \geq d_T$ holds. Then $\querypair$ returns $1$ with probability
	$1/2s_{m'-1}$, $-1$ with probability $1/2s_{m'-1}$, and $\bot$ otherwise.
	In the first two cases, the algorithm terminates with output $2m$ and
	$-2m$, respectively, so the outcomes cancel and the expected contribution
	to the estimator is $0$, as the algorithm continues running if $\bot$ is returned.
	\item $d_u^{\leq \dedge{uv}} < d_H, d_v^{\leq \dedge{uv}} < d_T$. Then
	$\querypair$ returns $\bot$ with probability $1$, and so the algorithm will not terminate at this step and the expected contribution to the estimator is 0.
\end{itemize}
This concludes the proof.
\end{proof}

\hedgevar*
\begin{proof}
Follows from the magnitude of the estimator being at most $2m$. 
\end{proof}

\begin{theorem}
There exists an quantum algorithm that approximates the number of edges
$\dedge{uv}$ with $d_u^{\leq \dedge{uv}} \geq d_H$ and $d_u^{\leq \dedge{uv}}
\geq d_H$ to $O(\eps m)$ error with probability $2/3$ and using $O(\frac{1}{\eps^2}
\log{n})$ quantum memory. This algorithm is classical except for the use of the quantum
sketch described in Section~\ref{sec:description}. 
\end{theorem}
\begin{proof}
By running $\bT{\frac{1}{\eps^2}}$ copies of the algorithm described above in
parallel and averaging the estimators, we can obtain an estimator with the
correct expectation, and variance $\eps^2/3$. The result then follows by
Chebyshev's inequality.
\end{proof}

\end{document}